\documentclass[12pt]{article}%
\usepackage{amssymb}
\usepackage{amsfonts}
\usepackage{amsmath}%
\usepackage{graphicx}
\usepackage{hyperref}
\usepackage{natbib}
\usepackage{float}

\usepackage{xargs}
\usepackage{setspace}
\usepackage[margin=1.25in]{geometry}

\usepackage{todonotes}

\usepackage{tikz}
\usetikzlibrary{patterns}
\usetikzlibrary{decorations.markings}

\newcommandx{\nima}[2][1=]{\todo[linecolor=black,backgroundcolor=black!25,bordercolor=black,#1]{#2}}
\newcommandx{\ron}[2][1=]{\todo[linecolor=red,backgroundcolor=red!25,bordercolor=red,#1]{#2}}

\newtheorem{theorem}{Theorem}

\newtheorem{definition}{Definition}
\newtheorem{example}{Example}

\newtheorem{lemma}{Lemma}

\newtheorem{proposition}{Proposition}

\usepackage{thm-restate}

\newenvironment{proof}[1][Proof]{\noindent\textbf{#1.} }{\ \rule{0.5em}{0.5em}}

\begin{document}

\title{A Theory of Stable Market Segmentations}
\author{Nima Haghpanah and Ron Siegel\thanks{Department of Economics, the Pennsylvania
State University, University Park, PA 16802 (e-mail: nuh47@psu.edu and
rus41@psu.edu).}}
\date{\today}
\maketitle

\begin{abstract}

We consider a monopolistic seller in a market that may be segmented. The surplus of each consumer in a segment depends on the price that the seller optimally charges, which depends on the set of consumers in the segment. We study which segmentations may result from the interaction among consumers and the seller. Instead of studying the interaction as a non-cooperative game, we take a reduced-form approach and introduce a notion of stability that any resulting segmentation must satisfy. A stable segmentation is one that, for any alternative segmentation, contains a segment of consumers that prefers the original segmentation to the alternative one. Our main result characterizes stable segmentations as efficient and saturated. A segmentation is saturated if no consumers can be shifted from a segment with a high price to a segment with a low price without the seller optimally increasing the low price. We use this characterization to constructively show that stable segmentations always exist. Even though stable segmentations are efficient, they need not maximize average consumer surplus, and segmentations that maximize average consumer surplus need not be stable. Finally, we relate our notion of stability to solution concepts from cooperative game theory and show that stable segmentations satisfy many of them.

\end{abstract}

\section{Introduction}

Sellers in a variety of markets use consumer data to conduct market segmentation and price discrimination. Various factors, including consumers' data disclosure decisions and other aspects of consumer behavior, determine which data sellers can access and, ultimately, the resulting market segmentation and prices consumers face. Consumers' decisions and their interaction with the seller may vary across settings, leading to different market segmentations. We are interested in understanding which market segmentations and prices would arise in different market settings.



We consider a market with a monopolistic seller of a single product and a continuum of consumers with unit demand, each characterized by their value for the product. We normalize the seller's cost of production to zero, and assume without loss of generality that consumer values are positive. Given a segmentation of the market, the seller sets the profit-maximizing price in each segment, which is the monopoly price given the distribution of consumer values in the segment.

If the seller has access to detailed consumer data and can unilaterally determine the segmentation, she will segment the market based on consumers' willingness to pay 
and implement first-degree price discrimination. If instead consumers determine the market segmentation before they learn their value for the product, they will rank the possible segmentations by the expected consumer surplus and choose a segmentation with the highest expected consumer surplus. Such segmentations were identified by \citet{BBM15}. But if consumers know their values for the product at the outset, different consumers may prefer different market segmentations. Which market segmentation would arise in this case?

To address this questions, we study the interaction among consumers and the seller in reduced form by focusing on properties of the resulting segmentation. We develop a notion of stability that corresponds to a segmentation being immune to deviations to other segmentations. The idea is that once a segmentation arises, any segment of consumers can object to transitioning to another segmentation, which prevents that segmentation from being realized. 
More precisely, a segmentation is a partition of the consumers into segments, where each segment is a set, or \emph{coalition}, of consumers together with a price for the coalition that maximizes the seller's profit. A segmentation $S$ is \emph{stable} if for any segmentation $S'$ that does not give all consumers the same surplus $S$ does,
there is some coalition $C$ of consumers in $S$ that objects to $S'$, that is, all consumers in $C$ weakly prefer $S$ to $S'$ and some consumers in $C$ strictly prefer $S$ to $S'$.
This notion of stability captures a kind of ``coalitional individual rationality (IR):'' once a segment is formed, its members cannot be regrouped into a different segment or segments if they all oppose this change.

Our main result characterizes stable segmentations. The characterization shows that a segmentation is stable if and only if it is \emph{efficient} and \emph{saturated}. Efficiency means that every consumer buys the product, and saturation means that consumers in each segment are not willing to accept additional consumers from a segment with a higher price because doing so will increase the price in their own segment. We also show that stable segmentations are Pareto undominated, that is, there is no other segmentation that makes all consumers better off. To clarify the separate roles that efficiency and saturation play in segmentation stability, we also characterize efficient segmentation. We show that a segmentation is efficient if and only if it is \emph{fragmentation-proof}, which means that it is immune to objections by coalitions of consumers that belong to the same segment.

   
We use our characterization of stable segmentations to show that they exist by constructing an efficient and saturated segmentation. This stable segmentation is the \emph{maximal equal-revenue segmentation}, identified by \citet{BBM15} as a segmentation that maximizes average consumer surplus among all segmentations. We then show that multiple stable segmentations may exist and that maximizing consumer surplus neither implies nor is implied by stability. 

Our results may be relevant to policy discussions regarding monopolies, price discrimination, and data sharing. Monopolies lead to inefficiencies, and these inefficiencies may be reduced with regulation or increased competition. Market segmentation arising from the monopolist's access to consumer data can also reduce inefficiency, but may harm consumers, as first-degree price discrimination demonstrates. Our results show that as long as consumers have enough control over their data to achieve ``coalitional IR,'' market segmentation leads to efficiency and a Pareto undominated outcome. This indicates that policies or information intermediaries that support coalitional IR while allowing the seller to use consumer data to price discriminate may offer an alternative or complimentary tool to addressing the inefficiencies associated with monopolistic markets.

Our model can also be described as a cooperative game with non-transferable utility (NTU). The set of players is the set of consumers. The set of feasible utility vectors associated with a measurable subset $C$ of players (coalition of consumers) consists of the profiles of payoffs of consumers in $C$ across all segmentations of $C$ (that is, when $C$ is taken to be the set of consumers).
A utility vector corresponding to some segmentation is in the core of this game if no coalition of consumers objects to the segmentation, that is, no coalition of consumers can obtain a weakly higher utility for all its members, and a strictly higher utility for some of its members. We show that the core of the game is non-empty if and only if it is profit-maximizing for the seller to set the efficient (lowest) price in the unsegmented market. In this case, the utility vector corresponding to the unsegmented market is essentially the unique element in the core. Moreover, any stable segmentation generates the same utility vector. Thus, if the core is not empty, it coincides with our notion of stability. In general, however, the core is empty but stable segmentations always exist.

Our notion of stability refines several solution concepts for NTU games applied to our setting. We show that a stable segmentation together with its equivalent segmentations, which give the same payoffs to the consumers, form a \citet{MoV53} stable set, a \citet{Har74} farsighted stable set, a \citet{RaV15} farsighted stable set, and a \citet{RaV19} maximal farsighted stable set.\footnote{Our notion of ``coalitional IR'' is closely related to the notion of ``coalitional sovereignty'' in \citet{RaV15} applied to our setting.}
We point out that stable sets and farsighted stable sets do not always exist in NTU games. And when they do exist they may necessarily includes multiple, non-equivalent utility vectors or partitions of players. The results above show that in our market game, ``singleton'' (up to equivalent segmentations) stable and farsighted stable sets always exist. In particular, for any deviation from a stable segmentation $S$ to another segmentation $S'$, there exists a path of ``credible'' and ``maximal'' (in the sense of \citealp{RaV19}) segmentations that leads back from $S'$ to $S$. This provides another justification for our notion of stability.

The rest of the paper is organized as follows. \autoref{sec: related literature} describes the relationship of our work to the existing literature. \autoref{sec: model} describes the model. \autoref{sec: characterization} introduces our notion of stability and the main results. \autoref{sec: cooperative} casts our model as a cooperative game and relates our notion of stability to cooperative solution concepts. \autoref{sec: conclusion} concludes.

\subsection{Related literature}\label{sec: related literature}


\citet{PeV21} consider stability of centralized markets against deviations by coalitions of agents. They show that fragmentation of such markets is unavoidable, despite its efficiency costs, except in special circumstances.   They study a bilateral trade setting, whereas we study a setting with a population of consumers and a seller. Further, when centralized markets are fragmented, \citet{PeV21} do not predict what the resulting segmentation looks like, whereas a characterization of stable segmentations is a main focus of our paper. Another important difference is that in their setting a coalition chooses the trading mechanism, whereas in our setting each coalition of consumers faces a profit-maximizing price set by the seller. 

A recent literature on third-degree price discrimination studies consumer and producer surplus across all possible segmentations of a given market.  \citet{BBM15} identify the set of producer and consumer surplus pairs that result from all segmentations of a given market. Their results also identify segmentations that maximize average consumer surplus.  \citet{CDH20} study an extension in which only certain segmentations may be chosen. \citet{Glo18} study optimal disclosure by an informed agent in a bilateral
trade setting, and show that the optimal disclosure policy leads to socially
efficient trade, even though information is revealed only partially.  \citet{Ich18},
\citet{HiV18}, \citet{Bra17}, and \citet{HaS21} consider maximum average consumer surplus when a multi-product seller
offers different products in each market segment.  These papers can be seen as identifying segmentations that are chosen \emph{ex ante} by a consumer who does not know her type, because such a consumer chooses the segmentation that maximizes her expected payoff.  In contrast, in this paper we study market segmentation when consumers know their type.

The related papers that study disclosure decisions by agents who know their type model these interactions as non-cooperative games. \citet{ALV20} consider voluntary disclosure of data by a single consumer, and analyze the welfare implications of various disclosure policies in both a monopolistic and a competitive environment.   \citet{ShV15} study a disclosure setting in which the seller can commit to the mechanism that she will use after receiving information.  In our setting, the seller cannot commit and chooses a profit-maximizing price for each segment.  \citet{AMM19} and \citet{BBG20} also study the consequences of consumers' disclosure decisions on prices and other market outcomes.  The main difference between these papers and ours is that we model the interaction between consumers and the seller in reduced form that can also be thought of as a cooperative game.

Our work also relates to the literature that endogenizes privacy costs because the payoff of each consumer depends on the set of other consumers in her segment.  In this literature \citep{Tay04,CaP06,CTW12,CLP16,BoC20,ArB21} privacy is valuable because of dynamic considerations: what an agent reveals today may be shared and used by other agents in the future.  In our setting, privacy is valuable because it allows a consumer to buy the product more cheaply by pooling with other consumers.

\section{Model\label{sec: model}}
A monopolistic seller faces a unit mass of consumers uniformly distributed on the interval $[0,1]$.  Consumers have unit demand for the monopolist's product, whose production cost is normalized to zero.  The value of the product for consumer $c \in [0,1]$ is $v(c) \in V = \{v_1,\ldots,v_n\} \subseteq R_{>0}$, where $v$ is a measurable function and $v_i$ increases in $i$.  Let $\mu$ be the Borel measure on the unit interval.  The measure of consumers with value $v_i$ is $f(v_i) = \mu(\{c: c \in [0,1], v(c) = v_i\})$.  We assume without loss of generality that $f(v_i) > 0$ for every $i \in \{1,\ldots,n\}$.

A \emph{coalition} is a measurable subset $C \subseteq [0,1]$ of consumers.  Let $f^C(v_i) = \mu(\{c: c \in C, v(c) = v_i\})$ denote the measure of consumers with value $v_i$ in coalition $C$.  We say that consumers with value $v_i$ are in $C$ (or that $C$ contains consumers with value $v_i$) if $f^C(v_i)>0$. A price $p \in V$ is \emph{optimal for coalition} $C$ if it maximizes the revenue from selling the product to consumers in $C$, that is, for any other price $p' \in V$,
\begin{align*}
    p \sum_{i: v_i \geq p} f^C(v_i) \geq p' \sum_{i: v_i \geq p'} f^C(v_i).
\end{align*}
We restrict attention to prices in $V$ because for any other price there exists a price in $V$ with a weakly higher revenue.  

A \emph{segment} is a pair $(C,p)$, where $C$ is a coalition and $p$ is an optimal price for that coalition. A \emph{segmentation} $S$ is a finite set of segments $\{(C_j,p_j)\}_{j=1,\ldots,k}$
such that $C_1,\ldots,C_k$ partitions the set of consumers $[0,1]$. That is, a segmentation partitions the set of consumers into coalitions, and assigns an optimal price for each coalition.  A segmentation is trivial if it consists of a single segment $\{([0,1],p)\}$. There may be multiple trivial segmentations because multiple prices may be optimal for the set $[0,1]$ of all consumers.

Denote by $CS(c,p) = \max\{v(c)-p,0\}$ the surplus of consumer $c$ who is offered the product at price $p$.  If consumer $c$ belongs to segment $(C,p)$, then her surplus is $CS(c,p)$. Given a segmentation $S$, denote by $p_S(c)$ the price in the unique segment that includes consumer $c$.  Let  $CS(c,S) = CS(c, p_S(c))$ denote the surplus of consumer $c$ in segmentation $S$.

Consumers' preferences over segmentations may differ because the prices different consumers face vary across segmentations. 
Consumers' data disclosure and other decisions, along with their interaction with the seller, which we do not explicitly model, determine the resulting segmentation. We model this process in reduced form by specifying a stability property for the resulting segmentation.


\section{Stable Segmentations}
We develop a notion of stable segmentations and show that such segmentations always exist. To this end we first formalize what it means for a segment to object to a segmentation.

\begin{definition}[Objection]
A segment $({C},p)$ objects to a segmentation $S$ if $CS(c,p) \geq CS(c,S)$ for all consumers $c$ in ${C}$, with a strict inequality for a positive measure of consumers $c$ in ${C}$. 
\end{definition}

In other words, a segment $({C},p)$ objects to a segmentation $S$ if all the consumers in $C$ are weakly worse off in $S$ and some consumers in $C$ are strictly worse off. In particular, $({C},p)$ is not in $S$.  Notice that the definition would be vacuous if we required the inequality to be strict for all (or almost all) consumers in $C$, because the optimality of $p$ for $C$ requires that the surplus of the consumers with the lowest value in $C$ is zero.

Next we formalize the notion of a blocking segmentation.

\begin{definition}[Blocking]
A segmentation $S'$ blocks a segmentation $S$ if there exists a segment $(C',M')$ in $S'$ that objects to $S$.
\end{definition}

We are now ready to introduce our notion of stability. For the definition, we say that segmentations $S$ and $S'$ are \emph{surplus-equivalent} if almost all consumers are indifferent between the two segmentations, that is, $CS(c,S)=CS(c,S')$ for almost every consumer $c$ in $[0,1]$.

\begin{definition}[Stability]
    A segmentation is stable if it blocks any segmentation that is not surplus-equivalent to it.
\end{definition}

Stability captures a kind of ``coalitional individual rationality (IR)'' in that no coalition can be forced to regroup into one or more different segments if all its members oppose this change. A segmentation is stable if no other segmentation can be reached without violating coalitional IR. Like other notions of stability, the definition does not specify the details of the interaction among the consumers and the seller and how a stable segmentation is reached. Instead, it can be thought of as ruling out deviations from any candidate segmentation that is the outcome of the unmodeled process. An alternative, equivalent definition is that a segmentation is stable if it counter-blocks any blocking segmentation. That is, no blocking segmentation can be reached without violating coalitional IR.\footnote{The equivalence of the two definitions follows from the characterization of stable segmentations in the next subsection because, using the terminology introduced there, a proof similar to that of \autoref{conjecture: stable with a single product} shows that the induced canonical segmentation of a segmentation that counter-blocks any blocking segmentation is efficient and saturated.}

The next sections show that stable segmentations exist and characterize them.  We start with the characterization because we use it to prove existence.



\subsection{Characterization of Stable Segmentation}\label{sec: characterization}

We start by defining the notion of a \emph{canonical} segmentation. A segmentation is canonical if no two segments have the same price.  For each segmentation $S$, all surplus-equivalent segmentations are surplus-equivalent to a unique canonical segmentation, which we call the \emph{induced canonical segmentation} of $S$. The induced canonical segmentation is obtained by merging all segments that have the same price into a single segment with that price.



Our characterization says that a segmentation is stable if and only if its induced canonical segmentation satisfies two properties: efficiency and saturation. A segmentation is \emph{efficient} if all consumers buy the product, that is, for any segment $(C,p)$ in the segmentation, the price $p$ is equal to the lowest value $\underline{v}(C):=\min\{v: f^C(v)>0\}$ in $C$. A segmentation is \emph{saturated} if for any segment $(C,p)$ in the segmentation, whenever we add consumers to coalition $C$ from a segment with a price strictly higher than $p$, price $p$ is sub-optimally low for this larger coalition.  That is, for any two segments $(C,p)$ and $(C',p')$ in the segmentation with $p<p'$ and any positive-measure set $C'' \subseteq C'$ of consumers, any optimal price for coalition $C \cup C''$ is strictly higher than $p$.

The following lemma shows that saturation can be expressed more succinctly by looking at the set of prices that are optimal for different segments.

\begin{lemma}\label{lem: succinct saturation}
A segmentation is saturated if and only if for any two segments $(C,p)$ and $(C',p')$ in the segmentation with $p<p'$, there exists a price $\hat{p}$ that is optimal for $C$ such that $p< \hat{p} \leq \underline{v}(C')$.
\end{lemma}
\begin{proof}If such a $\hat{p}$ exists, then by adding consumers from $C'$ to $C$, all of whose values are at least $\underline{v}(C')$, the revenue of price $\hat{p}$ increases more than the revenue of price $p$.  And because both $p$ and $\hat{p}$ are optimal for $C$, $p$ is no longer is optimal when we add these consumers. Conversely, if no such $\hat{p}$ exists, then $p$ is the highest optimal price for $C$ that does not exceed $\underline{v}(C')$, so if we add a small measure of consumers with value $\underline{v}(C')$ from $C'$ to $C$, price $p$ remains optimal for $C$.
\end{proof}



To prove the characterization of stable segmentations, we first relate stability to Pareto dominance. A segmentation $S$ \emph{Pareto dominates} another segmentation $S'$ if $CS(c,S) \geq CS(c,S')$ for all consumers $c$ in $[0,1]$, with a strict inequality for a positive measure of consumers. A segmentation $S$ is \emph{Pareto undominated} if there exists no segmentation $S'$ that Pareto dominates $S$.

\begin{lemma}\label{Prop: dominated then unstable}
If a segmentation is Pareto dominated, then it is not stable.
\end{lemma}
\begin{proof}
If $S$ is Pareto dominated by $S'$, then no segment in $S$ objects to $S'$, and $S'$ is not surplus-equivalent to $S$.  Therefore, $S$ is not stable.
\end{proof}

We now state and prove our main result.

\begin{theorem}\label{conjecture: stable with a single product}
        A segmentation is stable if and only if its induced canonical segmentation is efficient and saturated.
\end{theorem}
\begin{proof}
To see the necessity of efficiency, consider a segmentation $S'$ with an induced canonical segmentation $S$.  Suppose that $S$ is inefficient, so there is a segment $(C,p)$ in $S$ such that $p > \underline{v}(C)$.    Consider a coalition $\bar{C} \subseteq C$ that consists of the consumers in $C$ with values strictly lower than $p$ and a positive measure of the highest value consumers in $C$ that is small enough that any optimal price for $\bar{C}$ is strictly lower than $p$. Denote by $p'<p$ an optimal price for $\bar{C}$ so $(\bar{C},p')$ is a segment.  Observe that $p$ remains optimal for $C \backslash \bar{C}$. Indeed,  removing from $C$ consumers with values strictly lower than $p$, who do not purchase the product, does not change the revenue from $p$; and removing from $C$ some consumers with the highest value in $C$ can only lower the optimal price, but $p$ is already the lowest value of consumers in $C$ after removing the consumers with values lower than $p$, so $p$ remains optimal. 
Now consider a segmentation $\bar{S}$ obtained from segmentation $S$ by replacing segment $(C,p)$ with the two segments $(C \backslash \bar{C},p)$ and $(\bar{C},p')$.  The consumers in $\bar{C}$ with the highest value in $C$ have a strictly higher surplus in $\bar{S}$ than in $S$, and all the other consumers have a weakly higher surplus in $\bar{S}$ than in $S$. Thus, $\bar{S}$ Pareto dominates $S$ and therefore $S'$. By \autoref{Prop: dominated then unstable}, $S'$ is not stable.

To see the necessity of saturation, consider a segmentation $S'$ with an induced canonical segmentation $S$.  Suppose that $S$ is not saturated.  If $S$ is inefficient, then the argument above implies that $S'$ is not stable.  Suppose that $S$ is efficient, which together with non-saturation implies that there are two segments $(C,\underline{v}(C))$ and $(C',\underline{v}(C'))$ in $S$ with $\underline{v}(C)<\underline{v}(C')$ such that no $\hat{p}$ with $\underline{v}(C)< \hat{p} \leq \underline{v}(C')$ is optimal for $C$.  In particular, if we add a set $C''$ of consumers of value $\underline{v}(C')$ to $C$, price $\underline{v}(C)$ remains optimal provided that the measure of $C''$ is small. Let $C''$ be such a set that contains a positive measure of consumers with value $\underline{v}(C')$ from \emph{every} segment in $S'$ in which the price is $\underline{v}(C')$ (recall that $S'$ need not be canonical). Notice that $C''\subseteq C'$ because $S$ is the induced canonical segmentation of $S'$. Let $\bar{C}=C\cup C''$ and consider a segmentation $\bar{S}$ that is obtained from segmentation $S$ by replacing $(C,\underline{v}(C))$ and $(C',\underline{v}(C'))$ with $(\bar{C},\underline{v}(C))$ and $(C'\backslash \bar{C},p)$, where $p$ is any optimal price for $C'\backslash \bar{C}$.  Segmentations $S'$ and $\bar{S}$ are not surplus-equivalent because consumers in $C''$ have a surplus of zero in $S'$ and a positive surplus in $\bar{S}$.

We now argue that $S'$ does not block $\bar{S}$.  The surplus of consumers from any segment in $S'$ with a price different from $\underline{v}(C)$ nor $\underline{v}(C')$ does not change in $\bar{S}$, so these segments do not object to $\bar{S}$. Segments in $S'$ with price $\underline{v}(C)$ do not object to $\bar{S}$ because the consumers in these segments are in segment $(\bar{C},\underline{v}(C))$ in $\bar{S}$ and therefore face price $\underline{v}(C)$ in both segmentations. Finally, by construction of $C''$, for any segment $(C''',\underline{v}(C'))$ in $S'$, some of the value $\underline{v}(C')$ consumers, whose surplus is zero in $S'$, are in segment $(\bar{C},\underline{v}(C))$ and obtain a strictly positive surplus, so $(C''',\underline{v}(C'))$ does not object to $\bar{S}$.  We conclude that $S'$ does not block $\bar{S}$ so $S'$ is not stable.

We now turn to sufficiency.  Consider a segmentation $S'$ with an induced canonical segmentation $S$ that is efficient and saturated.  Let $\bar{S}$ be a segmentation that is not blocked by $S'$.  We will show that $\bar{S}$ is surplus-equivalent to $S'$. Since the canonical representation of $\bar{S}$ is also not blocked by $S'$ and is surplus equivalent to $S'$ if and only if $\bar{S}$ is surplus equivalent to $S'$, 
we suppose without loss of generality that $\bar{S}$ is canonical.  Write the two canonical segmentations as $S = \{(C_1,v_1),\ldots,(C_n,v_n)\}$ and $\bar{S} = \{(\bar{C}_1,v_1),\ldots,(\bar{C}_n,v_n\}$, where for each $i$ either $C_i$ is  empty or $v_i=\underline{v}(C_i)$ (because $S$ is efficient), and each $\bar{C}_i$ may be empty. We will show by induction that $C_i = \bar{C}_i$ for all $i$, which will prove that $\bar{S}$ is surplus-equivalent to $S'$.  (For the rest of this proof, $C_i = \bar{C}_i$ is in the ``almost all" sense, that is the measure of consumers in $C_i$ but not in $\bar{C}_i$ is zero, and the measure of consumers in $\bar{C}_i$ but not in $C_i$ is zero).

Suppose that $C_j = \bar{C}_j$ for all $j < i$ (the basis of the induction is $i=1$). We show that $C_i = \bar{C}_i$.  If $i=n$, then we are done because $S$ and $\bar{S}$ partition the same set $[0,1]$ of consumers.  Suppose that $i<n$. Since $C_j = \bar{C}_j$ for all $j < i$, a consumer faces a price $p \geq v_i$ in $S$ if and only if she faces a price $p' \geq v_i$ in $\bar{S}$. 
In particular, consumers in $C_i$ face a price of at least $v_i$ in $\bar{S}$.  Consumers in $C_i$ with values higher than $v_i$ must be in $\bar{C}_i$, otherwise these consumers face a price strictly higher than $v_i$ in $\bar{S}$,
so any segment in $S'$ that contains some of these consumers objects to $\bar{S}$: the consumers in this segment are in $C_i$ and face  price $v_i$ in $S'$ (because $S$ is the canonical representation of $S'$), and in $\bar{S}$ the consumers in this segment face prices no lower than $v_i$ (by the claim at the beginning of the paragraph).  So $C_i$ and $\bar{C}_i$ are identical, except that $\bar{C}_i$ may not contain some consumers of value $v_i$ from $C_i$ and may contain some consumers from coalitions $C_{i+1},\ldots,C_n$, all of whom have value strictly higher than $v_i$ (because $S$ is efficient). But, as we now argue, if $C_i$ and $\bar{C}_i$ are not identical, then the fact that $S$ is saturated contradicts the fact that $v_i$ is optimal for $\bar{S}$. To see this, suppose first that $\bar{C}_i$ does not contain some consumers of value $v_i$ from $C_i$. Since $S$ is saturated and $i<n$, by the succinct representation of saturation in \autoref{lem: succinct saturation} some price $p>v_i$ is optimal for $C_i$ and $p$ is smaller than the value of all consumers in coalitions $C_{i+1},\ldots,C_n$. Removing from $C_i$ some consumers with value $v_i$ reduces the revenue of price $v_i$ but not of price $p$, which makes price $v_i$ sub-optimal. Then, if needed, adding to $C_i$ consumers with values strictly higher than $v_i$ to obtain $\bar{C}_i$ makes price $v_i$ even worse (weakly) relative to price $p$, so $v_i$ is not optimal for $\bar{C}_i$. Now suppose that $\bar{C}_i$ differs from $C_i$ only because $\bar{C}_i$ contains some consumers from coalitions $C_{i+1},\ldots,C_n$, all of whom have value strictly higher than $v_i$. Adding these consumers to $C_i$ makes price $v_i$ sub-optimal because $S$ is saturated, so $(\bar{C}_i,v_i)$ is not a segment, a contradiction.
\end{proof}

We underscore that to verify the stability of a segmentation, efficiency and saturation must be checked for its \emph{canonical} representation.  The following example describes a non-canonical segmentation that is efficient and saturated but not stable.

\begin{example}\label{example 2}
        There are four values, $1$ to $4$, with measures $0.5,0.25,0.125,0.125$, respectively, as shown in \autoref{fig: example 2}.
\begin{figure}
    \centering
    		\begin{tikzpicture}[scale=4.5, ultra thick]
		\draw (0,0) node[left]{Consumers} -- (1,0);
		\draw (0,.05) -- (0,-0.05) node[below]{$0$};
		\draw (1,.05) -- (1,-0.05) node[below]{$1$};
		\draw (0,.25) node[left]{Values};
		\draw[loosely dotted, thick] (.5,0) node[below, yshift=-3]{$\tfrac{1}{2}$} -- (.5,.25);
		\draw[loosely dotted, thick] (.75,0) node[below, yshift=-3]{$\tfrac{3}{4}$} -- (.75,.25);
		\draw[loosely dotted, thick] (.875,0) node[below, yshift=-3]{$\tfrac{7}{8}$} -- (.875,.25);
		\draw (.2,0.25) node{$1$};
		\draw (.625,0.25) node{$2$};
		\draw (.8125,0.25) node{$3$};
		\draw (.9375,0.25) node{$4$};
		
		\draw (0,-.3) node[left]{$C_1$};
		\draw (0.25,-.3) -- (.75,-.3);
		\draw (.25,-.25) -- (.25,-0.35) node[below]{$\tfrac{1}{4}$};
		\draw (.74,-.255) arc (40:-40:0.07) node[below]{$\tfrac{3}{4}$};		
		
		
		\draw (0,-.6) node[left]{ $C_2$};
		\draw (.75,-.6) -- (.875,-.6);
		\draw (.865,-.555) arc (40:-40:0.07) node[below]{$\tfrac{7}{8}$};
		\draw(.75,-.55) -- (.75,-.65) node[below]{$\tfrac{3}{4}$};
		
		\draw (0,-.6) -- (.25,-.6);
		\draw (.24,-.555) arc (40:-40:0.07) node[below]{$\tfrac{1}{4}$};
		
		\draw(0,-.55) -- (0,-.65) node[below]{$0$};	
		
		\draw (0,-.9) node[left]{ $C_3$};
		\draw (.875,-.9) -- (1,-.9);
		\draw (.875,-.85) -- (.875,-.95) node[below]{$\tfrac{7}{8}$};
		\draw(1,-.85) -- (1,-.95) node[below]{$1$};		
		\end{tikzpicture}
    \caption{\autoref{example 2}.}
\label{fig: example 2}
\end{figure}
\end{example}
Consider the segmentation $S=\{(C_1,1), (C_2,1), (C_3,4)\}$ with coalitions $C_1,C_2,C_3$, shown in \autoref{fig: example 2}.

For coalition $C_1$, prices $1$ and $2$ are optimal.  For coalition $C_2$, prices $1$ and $3$ are optimal.  Adding consumers from segment $(C_3,4)$ to either segment $(C_1,1)$ or $(C_2,1)$ necessarily increases the optimal price in the latter segments. Thus, the segmentation is saturated.  The segmentation is also clearly efficient.

Notice, however, that for coalition $C_1 \cup C_2 = [0,\tfrac{7}{8})$, price $1$ is the unique optimal price, so the segmentation ${(C_1 \cup C_2,1),(C_3,4)}$ is efficient but not saturated. We can add some consumers with value 4 to $C_1 \cup C_2$ without changing the optimal price.  To see that $S$ is not stable, note that segmentation $S' = \{([0,\tfrac{7}{8}+\epsilon),1),([\tfrac{7}{8}+\epsilon,1],4)\}$ for some small $\epsilon > 0$ Pareto dominates $S$, so $S$ is not stable by \autoref{Prop: dominated then unstable}.

\subsection{Existence of Stable Segmentations}
We use our characterization of stable segmentations to construct a stable segmentation. This proves that stable segmentations always exist.  We start with an example that demonstrates the construction.

\begin{example}[Maximal Equal-revenue Segmentation]\label{example equal revenue}
        There are three values. $1,2,3$, with measures $\tfrac{1}{3},\tfrac{1}{6},\tfrac{1}{2}$, respectively, as shown in \autoref{figure 3}.
\begin{figure}
    \centering
    	\begin{tikzpicture}[scale=4.5, ultra thick]
		\draw (0,0) node[left]{Consumers} -- (1,0);
		\draw (0,.05) -- (0,-0.05) node[below]{$0$};
		\draw (1,.05) -- (1,-0.05) node[below]{$1$};
		\draw (0,.25) node[left]{Values};
		
		\draw[loosely dotted, thick] (1/3,0) node[below, yshift=-3]{$\frac{1}{3}$} -- (1/3,.25);
		\draw[loosely dotted, thick] (1/2,0) node[below, yshift=-3]{$\frac{1}{2}$} -- (1/2,.25);
		
		\draw (1/6,0.25) node{$1$};
		\draw (2.5/6,0.25) node{$2$};
		\draw (7/9,0.25) node{$3$};
		
		
			\draw (0,-.3) node[left]{$C_1$};
			\draw (0,-.3) -- (4/9,-.3);
			\draw (0,-.25) -- (0,-0.35) node[below]{$0$};
			\draw (4/9-0.01,-.25) arc (40:-40:0.07) node[below]{$\frac{4}{9}$};

			\draw (7/9,-.3) -- (1,-.3);
			\draw (7/9,-.25) -- (7/9,-0.35) node[below]{$\frac{7}{9}$};
			\draw (1,-.25) -- (1,-0.35) node[below]{$1$};							
			\draw (0,-.6) node[left]{$C_2$};
			\draw (4/9,-.6) -- (11/18,-.6);
			\draw (11/18-0.01,-.55) arc (40:-40:0.07) node[below]{$\frac{11}{18}$};

			\draw (4/9,-.55) --(4/9,-.65) node[below]{$\frac{4}{9}$};

			\draw (0,-.9) node[left]{$C_3$};
\draw (11/18,-.9) -- (7/9,-.9);
\draw (11/18,-.85) -- (11/18,-.95) node[below]{$\frac{11}{18}$};

\draw (7/9-0.01,-.85) arc (40:-40:0.07) node[below]{$\frac{7}{9}$};		
		\end{tikzpicture}
		\caption{\autoref{example equal revenue}.}
		\label{figure 3}
\end{figure}
\end{example}
    
    Consider the segmentation $S = \{(C_1,1),(C_2,2),(C_3,3)\}$ shown in \autoref{figure 3}.  Coalition $C_1$ is the largest ``equal-revenue" coalition that includes all values.  That is, the measures $\frac{3}{9}, \frac{1}{9}, \frac{2}{9}$ of the three values in coalition $C_1$ are such that prices $1$, $2$, and $3$ are all optimal, and $C_1$ is the largest such coalition because it contains all the consumers with value 1.\footnote{Any two coalitions $C,C'$ for which all three prices are optimal are proportional, that is, $f^C(v) = \alpha f^{C'}(v)$ for some $\alpha>0$, so the largest such coalition is well-defined.}   Segment $(C_1,1)$ is efficient.  Consumers not in $C_1$ have value either 2 or 3, so adding them to $C_1$ makes price $1$ no longer optimal.
    
    Having put all the consumers with value $1$ in segment $C_1$, we define the rest of the segmentation recursively to guarantee efficiency and saturation.  The values of the remaining consumers are $2$ and $3$, and the measures of consumers with these values are $\frac{1}{18}$ and $\frac{5}{18}$, respectively. Coalition $C_2$, in which values 2 and 3 have measures $\frac{1}{18}$ and $\frac{2}{18}$, is the largest coalition for which prices $2$ and $3$ are both optimal.  The segment $(C_2,2)$ is efficient, and adding any of the remaining consumers, all of whom have value 3, increases the optimal price. The last segment is $(C_3,3)$, which is efficient.  Thus, segmentation $S$ is efficient and saturated.  Because it is also canonical, it is stable.

We now formally define the maximal equal-revenue segmentation.  Let $\bar{F}^C(v_i)$ be the cumulative measure of consumers with values $v_i$ or higher in coalition $C$. If $v_i \bar{F}^C(v_i) = v_j \bar{F}^C(v_j)$, then prices $v_i$ and $v_j$ generate the same revenue for coalition $C$.  Coalition $C$ is an equal-revenue coalition if all consumer values in the coalition generate the same revenue, that is, $v_i\bar{F}^C(v_i)$ is the same for all $v_i$ with $f^C(v_i) > 0$.  A maximal equal-revenue segmentation is defined recursively. The first coalition, $C_1$, is the largest equal-revenue coalition that includes all the values. To construct $C_1$, let $\lambda_1$ be the eventual revenue in coalition $C_1$ from each of the values, that is, $\lambda_1 = v_i\bar{F}^{C_1}(v_i)$ for all $v_i$ in $V$. Recalling that $f^{C_1}(v_i)$ is the measure of consumers with value $v_i$ in coalition $C$, and $f(v_i)$ is the overall measure of consumers with value $v_i$, we have that $f^{C_1}(v_i) =\bar{F}^{C_1}(v_i) - \bar{F}^{C_1}(v_{i+1})= \lambda_1(\frac{1}{v_i}-\frac{1}{v_{i+1}}) \leq f(v_i)$ for all $v_i$, where $\frac{1}{v_{n+1}}\equiv0$. Therefore, the highest value that $\lambda_1$ can take is such that $f^{C_1}(v_i) =f(v_i)$ for some $i$. That is, $\lambda_1$ is the smallest value such that $\lambda_1(\frac{1}{v_i}-\frac{1}{v_{i+1}}) = f(v_i)$ for some $i$. Denote the index of this value by $i_1$, so $\lambda_1(\frac{1}{v_{i_1}}-\frac{1}{v_{i_1+1}}) = f(v_{i_1})$. Then, more succinctly, we define $C_1$ by letting
\begin{align}
    \lambda_1 = \min_{v_i \in V} \frac{f(v_i)}{\frac{1}{v_i}-\frac{1}{v_{i+1}}}=\frac{f(v_{i_1})}{\frac{1}{v_{i_1}}-\frac{1}{v_{i_1+1}}},\label{eq: equal revenue definition largest revenue}
\end{align}
and letting $\bar{F}^{C_1}(v_i) = \lambda_1/v_i$ for all $v_i$.


Coalition $C_1$ contains all consumers with value $v_{i_1}$, and adding a positive measure of consumers with other values to $C_1$ makes price $v_{i_1}$ sub-optimal.  Therefore coalition $C_1$ cannot be any larger and still be an equal-revenue coalition.  The first segment in a maximal equal-revenue segmentation is $(C_1,v_1)$, and the rest of the segmentation is defined recursively, where the coalition $C_j$ in the $j$'th segment is the largest equal-revenue coalition that includes all the values that remain after removing the consumers in $C_1,\ldots,C_{j-1}$, that is, $\{v_i: f^{C_j}(v_i)>0\} = \{v_i: f^{C \backslash \cup_{j'<j}C_{j'}}(v_i)>0\}$, and the price in the $j$'th segment is $\min\{v_i: f^{C_j}(v_i)>0\}$.  This process ends because in each step the number of remaining values decreases by at least 1.

The maximal equal-revenue segmentation is not necessarily canonical.  For example, if the first equal-revenue coalition $C_1$ exhausts some value other than $v_1$, then the second coalition, $C_2$, will also include consumers with value $v_1$.  By \autoref{conjecture: stable with a single product}, to establish that the maximal equal-revenue segmentation is stable, we need show that its induced canonical segmentation is efficient and saturated.

\begin{proposition}
    The maximal equal-revenue segmentation is stable.
\end{proposition}
\begin{proof}
The price in each segment of the maximal equal-revenue segmentation is equal to the lowest consumer value in the segment, so the segmentation is efficient and the same is true for its induced canonical segmentation.  It remains to show that the induced canonical segmentation is saturated.

By construction of the maximal equal-revenue segmentation, for any two segments $(C_i,v_i)$ and $(C_j,v_j)$ with $i<j$, the set of consumer values in $C_j$ is a subset of that in $C_i$.
Since coalition $C_j$ contains consumers with value $v_j$ ($f^{C_j}(v_j)>0$), so does coalition $C_i$. Because $C_i$ is an equal-revenue segmentation, price $v_j$ is optimal for coalition $C_i$. 
Consider the segment with price $v_i$ in the canonical segmentation. By definition of the induced canonical segmentation, the coalition in this segment is the union of all the coalitions with price $v_i$ in the maximal equal-revenue segmentation. Price $v_j$ is optimal for each of these coalitions, as argued in the previous paragraph, and is therefore optimal for the union of these coalitions. Thus, the induced canonical segmentation is saturated by \autoref{lem: succinct saturation}.
\end{proof}

The maximal equal-revenue segmentation is not the unique stable segmentation.  Here is an informal description of another construction of a stable segmentation.  Put all consumers of value $v_1$ in the first coalition, and continually add consumers with the lowest remaining value to the first coalition until some price $v_i$ other than $v_1$ also becomes  optimal.   This forms the first coalition, $C_1$. The first segment is $(C,v_1)$.  Repeat this process with the remaining consumers (the last segment may have only one optimal price).  The resulting segmentation is canonical, efficient, and saturated.  Saturation follows because given a segment $(C,v_j)$ so constructed, a value $v_k>v_j$ becomes optimal for coalition $C$ only when we have already added all the available consumers with values lower than $v_k$ to $C$, so the value of any consumer in a segment with a higher price is at least $v_k$, and adding such consumers to $C$ makes price $v_j$ sub-optimal.  This segmentation is also typically different from the maximal equal-revenue segmentation because the first segment does not generally include all values.  We proved the existence of stable segmentations by constructing the maximal equal-revenue segmentation because it maximizes average consumer surplus, as we discuss next.

\subsection{Stability vs. Maximizing Average Consumer Surplus}
The maximal equal-revenue segmentation was first introduced by \citet{BBM15}.  They showed that this segmentation maximizes average consumer surplus across all segmentations, but is not necessarily the only segmentation that does so.  What is the relationship between stability and maximization of average consumer surplus? The following two examples show that stability is neither necessary nor sufficient for maximization of average consumer surplus.

\begin{example}
        \textbf{\emph{(A segmentation that maximizes average consumer surplus and is not stable)}}\label{example max CS not stable} Consider again the example from \autoref{example equal revenue} with three values, $1,2,3$, and measures $\tfrac{1}{3},\tfrac{1}{6},\tfrac{1}{2}$, and segmentation $S = \{(C_1,1),(C_2,2)\}$ with coalitions $C_1 = [0,\tfrac{1}{3}] \cup [\tfrac{5}{6},1]$ and $C_2 = (\tfrac{1}{3},\tfrac{5}{6})$.
\end{example}        
        Coalition $C_1$ contains all value 1 consumers and some value 3 consumers in a proportion that makes prices $1$ and $3$ optimal. Coalition $C_2$ contains remaining consumers, whose proportions are such that prices $2$ and $3$ are optimal.  Segmentation $S$ maximizes average consumer surplus across all segmentations.\footnote{The segmentation is efficient so it maximizes total surplus.  It also minimizes the seller's revenue across all segmentations.  This is because the same price that is optimal for the set of all consumers, 3, is also optimal for each coalition.}
        
        But the segmentation is not stable.  Since it is canonical and efficient, to show that it is not stable, we show that it is not saturated.  Indeed, adding a small measure $\epsilon>0$ of value 2 consumers to $C_1$ does not make price $1$ sub-optimal: price $2$ is not optimal for $C_1$, so if $\epsilon$ is small enough, price $2$ remains sub-optimal, and the addition increases the revenue from price $1$ but does not change the revenue from price $3$.

\begin{example}\label{example 4}
\textbf{\emph{(A stable segmentation that does not maximize average consumer surplus)}}
        There are three values, $1,2,3$, each with measure $\tfrac{1}{3}$,  as shown in \autoref{fig:example 4}.
\begin{figure}
    \centering
    		\begin{tikzpicture}[scale=4.5, ultra thick]

	\draw (0,0) node[left]{Consumers} -- (1,0);
	\draw (0,.05) -- (0,-0.05) node[below]{$0$};
	\draw (1,.05) -- (1,-0.05) node[below]{$1$};
	\draw (0,.25) node[left]{Values};
	
	\draw[loosely dotted, thick] (1/3,0) node[below, yshift=-3]{\footnotesize $\frac{1}{3}$} -- (1/3,.25);
	\draw[loosely dotted, thick] (2/3,0) node[below, yshift=-3]{\footnotesize $\frac{2}{3}$} -- (2/3,.25);
	
	\draw (1/6,0.25) node{$1$};
	\draw (1/2,0.25) node{$2$};
	\draw (5/6,0.25) node{$3$};
	
	
	\draw (0,-.3) node[left]{$C_1$};
	\draw (0,-.3) -- (2/3,-.3);
	\draw (0,-.25) -- (0,-0.35) node[below]{$0$};
	\draw (2/3-0.01,-.255) arc (40:-40:0.07) node[below]{\footnotesize $\frac{2}{3}$};								
	
	\draw (0,-.6) node[left]{$C_2$};
	\draw (2/3,-.6) -- (1,-.6);
	\draw (1,-.55) -- (1,-.65) node[below]{$1$};
	
	\draw (2/3,-.55) -- (2/3,-.65) node[below]{\footnotesize $\frac{2}{3}$};

		\end{tikzpicture}
		\caption{\autoref{example 4}}
		\label{fig:example 4}
\end{figure}
\end{example}

Consider the segmentation $S = \{(C_1,1),(C_2,3)\}$ in \autoref{fig:example 4}. Coalition  $C_1$ consists of the consumers with values 1 and 2; Coalition $C_3$ consists of the consumers with value 3.  This segmentation is canonical, efficient, and saturated.  It is therefore stable.

The surplus of this segmentation is $\tfrac{1}{3}$, whereas the surplus of the maximal equal-revenue segmentation is $\tfrac{2}{3}$.\footnote{The maximal equal-revenue segmentation is $\{(C''_1,1),(C''_2,2)\}$ where $(f^{C''_1}(1),f^{C''_1}(2),f^{C''_1}(3)) = (\tfrac{1}{3},\tfrac{1}{9},\tfrac{2}{9})$ and $(f^{C''_2}(1),f^{C''_2}(2),f^{C''_2}(3)) = (0,\tfrac{2}{9},\tfrac{1}{9})$.  The surplus of value 2 and 3 consumers is $1$ and $2$ in the first segment, and the surplus of value 3 consumers is $1$ in the second segment.  The average consumer surplus is therefore $\tfrac{1}{9}\cdot 1 + \tfrac{2}{9}\cdot 2 + \tfrac{1}{9}\cdot 1 = \tfrac{2}{3}$.}  For some intuition, it is illuminating to study the marginal improvement in the average consumer surplus of $S$ obtained by swapping the same measure of value 2 and 3 consumers.  For this, consider a coalition $C'_1$ obtained from $C_1$ by removing measure $\epsilon$ of value 2 consumers and adding a measure $\epsilon$ of value 3 consumers, and a coalition $C'_2$ that contains the remaining consumers.  If $\epsilon > 0$ is small enough, price $1$ is optimal for $C'_1$ and price $3$ is optimal for $C'_3$, so $S' = \{(C'_1,1),(C'_3,3)\}$ is a segmentation.  To compare the  average consumer surplus of $S$ and $S'$, it suffices to consider the swapped value 2 and 3 consumers.  Each value 2 consumer loses 1 unit of surplus: their surplus is $1$ in $S$ and $0$ in $S'$. Each value 3 consumer gains 2 units of surplus each: their surplus is $0$ in $S$ and $2$ in $S'$.\footnote{Notice that the change in the offered price is $2$ for the value 2 consumers (from $1$ to $3$) and $-2$ for the value 3 consumers (from $3$ to $1$).  But even though this change has the same absolute value, the surplus change for the value 3 consumers is higher than the value 2 consumers because value 2 consumers do not buy the product at a price higher than 2 (so increasing the price they face from 2 to 3 does not change their surplus.)}

\subsection{Pareto Dominance and Efficiency}
Recall from \autoref{sec: characterization} that any stable segmentation is Pareto undominated and efficient.  In fact, any Pareto undominated segmentation is efficient, as the following result shows.

\begin{proposition}\label{prop: Pareto undominated then efficient}
    Any Pareto undominated segmentation is efficient.
\end{proposition}
\begin{proof}
Consider any inefficient segmentation $S$.  We show that $S$ is Pareto dominated by some other segmentation. By definition there must be a segment $(C,p)$ in $S$ such that $p$ is higher than the lowest value $\underline{v}(C)$ in $C$.  Construct a coalition $C' \subseteq C$ that contains all the consumers of lowest value $\underline{v}(C)$ from $C$, and an $\epsilon$ fraction of consumers of each other value from $C$.  If $\epsilon$ is small enough, price $\underline{v}(C)$ is optimal in $C'$, so $(C', \underline{v}(C))$ is a segment.  Also, price $p$ remains optimal in the remaining segment $C \backslash C'$. To see this, first notice that $\underline{v}(C)$ is lower than all values in $C\backslash C'$, and so price $\underline{v}(C)$ is not optimal for $C \backslash C'$. Second, for any $v_i > \underline{v}(C)$, we have $f^{C \backslash C'}(v_i) = (1-\epsilon)f^C(v_i)$. Therefore the revenue of any price $p' \in V \backslash \{\underline{v}(C)\}$ in $C \backslash C'$ is 
\begin{align*}
    p' \sum_{v_i \geq p'} f^{C \backslash C'}(v_i) = (1-\epsilon) (p' \sum_{v_i \geq p'} f^{C}(v_i)),
\end{align*}
which is maximized at $p' = p$ since price $p$ is optimal in $C$.  

Now consider a segmentation $S'$ that is identical to $S$, except the segment $(C,p)$ is replaced with two segments $(C',\underline{v}(C))$ and $(C \backslash C', p)$.  Since the price in the segment $(C',\underline{v}(C))$ is lower than that in $(C,p)$, all consumers other than those with value $\underline{v}(C)$ are strictly better off in $C'$ relative to $C$, and the consumers with value $\underline{v}(C)$ have zero surplus in either case.  Further, since the prices in segments $(C \backslash C', p)$ and $(C,p)$ are the same, the surplus of each consumer in $C \backslash C'$ remains unchanged. Therefore $S'$ Pareto dominates $S$.
\end{proof}

There exist efficient segmentations that are not Pareto undominated.  For instance, the segmentation that puts all the consumers of each value in a different segment is efficient but is not Pareto undominated because it gives zero surplus to each consumer. This segmentation is not stable, because, as we have shown, stable segmentations are Pareto undominated.

There also exist Pareto undominated segmentations that are not stable. Such segmentations have induced canonical segmentations that are not saturated, since Pareto undominated segmentations are efficient. In order for an efficient segmentation that is not saturated to be Pareto undominated, it has to be that when consumers are made better off by moving from a segment with a higher price to an unsaturated segment with a lower price without increasing the latter segment's price, the price in the segment with the higher price necessarily increases. This is the case in \autoref{example max CS not stable}, which describes a segmentation that maximizes average consumer surplus, and is therefore Pareto undominated and efficient (by \citealp{BBM15}), but is not stable.

\subsection{Fragmentation-Proofness and Efficiency }

We have seen that stability is equivalent to saturation and efficiency of the induced canonical segmentation. We also expressed saturation by looking at the set of optimal prices for different segments. To further understand the interaction between saturation and efficiency, we now describe efficient segmentations as segmentations that are immune to certain objections. To this end we define \emph{fragmentation-proofness}, which excludes objections by coalitions that include consumers from more than one segment.

\begin{definition}
   A segmentation $S$ is fragmentation-proof if there exists no objection $(C,p)$ to $S$ such that $C \subseteq C'$ for some segment $(C',p')$ in $S$.
\end{definition}

A fragmentation-proof segmentation does not have an objection by any coalition that is a subset of consumers in an existing coalition. We show that fragmentation-proofness is equivalent to efficiency.

\begin{proposition}
    A segmentation is fragmentation-proof if and only if it is efficient.
\end{proposition}
\begin{proof}
Consider an efficient segmentation and a segment $(C,p)$ in the segmentation.  Because the segmentation is efficient, $p=\underline{v}(C)$ is the lowest consumer value in $C$.  Thus, the optimal price for any subset of $C$ is at least $p$ and there exists no objecting segment $(C',p')$ with $C' \subseteq C$.

Now consider an inefficient segmentation and a segment $(C,p)$ in the segmentation such that $p>\underline{v}(C)$.  Consider a coalition $C' \subseteq C$ that contains all the value $\underline{v}(C)$ consumers in $C$ and a small enough measure of the consumers with the highest value in $C$ so that the unique optimal price for $C'$ is $\underline{v}(C)$. The segment $(C',\underline{v}(C))$ objects to $S$, so $S$ is not fragmentation-proof.
\end{proof}


\subsection{Environments with Two Values}
We have seen that, in general, stability, maximizing average consumer surplus, and Pareto undominance are different concepts. When there are only two consumer values, however, these concepts coincide.  Moreover, there is essentially a single segmentation that satisfies these properties, in a sense slightly weaker than surplus-equivalence.  Formally,
two segmentations are \emph{weakly surplus-equivalent} if each of the two corresponding induced canonical segmentations can be obtained from the other via a measure-preserving mapping that, for each value, maps the set of consumers with that value to itself. More precisely, for any two segments $(C,p)$ and $(C',p)$ with the same price in the two canonical segmentations, $f^C(v) = f^{C'}(v)$ for all $v$. Clearly, any two surplus-equivalent segmentations are weakly surplus-equivalent because surplus-equivalence requires that $C =C'$.

\begin{proposition}\label{proposition: characterization and existence with two types}
Suppose that there are only two values.  For any segmentation $S$, the following three statements are equivalent:
\begin{enumerate}
    \item $S$ is stable.
    \item $S$ is Pareto undominated.
    \item $S$ maximizes average consumer surplus.
\end{enumerate}
 Segmentations that satisfy these three equivalent properties are weakly surplus-equivalent.
\end{proposition}
\begin{proof}
Suppose first that $v_1$ is optimal for the coalition $[0,1]$.  Then the trivial segmentation $\{([0,1]),v_1\}$ gives the highest possible surplus to all consumers, so a segmentation is Pareto undominated if and only if it maximizes average consumer surplus if and only if it is surplus-equivalent (and therefore weakly surplus-equivalent) to this segmentation.  We will argue in \autoref{prop: core empty then = stable} below that if $v_1$ is optimal for $[0,1]$, then all stable segmentations are surplus-equivalent to $\{([0,1]),v_1\}$.

Now suppose that $v_1$ is not optimal for $[0,1]$, that is, $v_1(f^{[0,1]}(v_1)+f^{[0,1]}(v_2)) < v_2f^{[0,1]}(v_2)$.  Consider any segmentation $S = \{(C_1,v_1),(C_2,v_2)\}$ such that $C_1$ contains all the value 1 consumers a measure of value 2 consumers so that $v_1(f^{C_1}(v_1)+f^{C_1}(v_2)) = v_2f^{C_1}(v_2)$, and $C_2$ contains the remaining value $2$ consumers.  We show that any segmentation that satisfies either of the three properties, stability, Pareto undominance, and maximizing average consumer surplus, is weakly surplus-equivalent to $S$.

Consider the induced canonical segmentation $S'' = \{(C''_1,v_1),(C''_2,v_2)\}$ of some segmentation $S'$.  The surplus of value $v_2$ consumers in $C''_1$ is $v_2-v_1$, and the surplus of all other consumers is zero.  So $S'$ maximizes average consumer surplus if and only if $C''_1$ has the maximal possible measure of value $v_2$ consumers, that is, if and only if it is Pareto undominated.
And $C''_1$ has the maximal possible measure of value $v_2$ consumers if and only if $f^{C''_2}(v_1) = 0$ and  $v_1(f^{C''_1}(v_1)+f^{C''_1}(v_2)) = v_2f^{C''_1}(v_2)$.\footnote{Indeed, if $f^{C''_2}(v_1) > 0$, then we can add some consumers of value $v_1$ and $v_2$ from $C''_2$ to $C''_1$; and if $v_1(f^{C''_1}(v_1)+f^{C''_1}(v_2)) > v_2f^{C''_1}(v_2)$, then we can add some consumers with value $v_2$ from $C''_2$ to $C''_1$.} That is, of and only if $S'$ is weakly surplus-equivalent to $S$.  Also, $f^{C''_2}(v_1) = 0$ and $v_1(f^{C''_1}(v_1)+f^{C''_1}(v_2)) = v_2f^{C''_1}(v_2)$ mean that $S''$ is saturated and efficient so $S'$ is stable.
\end{proof}

\section {Relationship to Cooperative Game Theory}\label{sec: cooperative}

Our model can be described as a cooperative game with non-transferable
utility (NTU). The players are the consumers. For each coalition $C$ of consumers, the set of utility vectors feasible for $C$ comprises the payoff profiles of the consumers in $C$ across all segmentations of $C$ (so $C$, instead of $[0,1]$,  is taken to be the set of consumers).\footnote{The continuum of players in our settings requires minor adjustments to the cooperative solution concepts we discuss, which are typically defined for games with a finite number of players.}
We first examine the core of the game and relate it to our notion of stability. We then relate stability to several other solution concepts for NTU games.

\subsection{The core}\label{sec: core}

We define the core to be the set of segmentations to which there is no objection.\footnote{For our purposes it is more convenient to refer to a set of segmentations instead of the payoff vectors they induce.}

\begin{definition}[Core]
    The core is a set of segmentations.  A segmentation $S$ is in the core if there exists no segment that objects to $S$.
\end{definition}

The core is a demanding solution concept. We characterize when the core is not empty and show that when the core is not empty it contains an essentially unique segmentation.

\begin{proposition}\label{thm: core iff efficient single product}
If the market is efficient, that is, price $v_1$ is optimal for the set $[0,1]$ of all consumers, then the core consists of all segmentations that are surplus-equivalent to the trivial segmentation $\{([0,1],v_1)\}$.  Otherwise, the core is empty.
\end{proposition}
\begin{proof}
Suppose the market is efficient.
Then  $([0,1],v_1)$ is a segment and there is no objection to the trivial segmentation $\{([0,1],v_1)\}$ because in any segment $(C,p)$ the price is at least $v_1$.  For the same reason, any segmentation that is surplus-equivalent to $\{([0,1],v_1)\}$ is also in the core.\footnote{There are infinitely many  segmentations that are surplus-equivalent to the trivial segmentation.  For example, we can divide $[0,1]$ into two coalitions $C_1,C_2$ such that the relative measure of all the values is the same in $C_1,C_2$, and $[0,1]$, so $\{(C_1,v_1),(C_2,v_1)\}$ is a surplus-equivalent segmentation.}  Now consider a segmentation $S$ that is not surplus-equivalent to $\{([0,1],v_1)\}$, which means that the price $v$ in some segment is higher than $v_1$. Because $v$ is optimal for the coalition in the segment, the coalition contains a positive measure of consumers with values $v$. And because the price in any segment is at least $v_1$, the segment $([0,1],v_1)$ objects to $S$, so $S$ is not in the core.  We conclude that if the market is efficient, then
the core consists of all segmentations that are surplus-equivalent to the trivial segmentation $\{([0,1],v_1)\}$.

Now suppose that the market is not efficient, so $v_1$ is not optimal for the set $[0,1]$ of all consumers.  We first claim that any segmentation includes a segment with a price strictly higher than $v_1$.  Suppose for contradiction that there exists a segmentation $S = \{(C_j,v_1)\}_{j=1,\ldots,k}$.  Because $v_1$ is optimal for $C_j$, for any $p' \in V$ we have
\begin{align*}
    v_1 \sum_{i: v_i \geq v_1} f^{C_j}(v_i) \geq p' \sum_{i: v_i \geq p'} f^{C_j}(v_i).
\end{align*}
Because the segments $C_1,\ldots,C_k$ partition $[0,1]$, we have $\sum_{j=1}^k f^{C_j}(v_i) = f(v_i)$ for each $v_i$.  Summing the above inequality over all $j$, we have
\begin{align*}
    v_1 \sum_{i: v_i \geq v_1} f(v_i) \geq p' \sum_{i: v_i \geq p'} f(v_i).
\end{align*}
This is a contradiction to the assumption that $v_1$ is not optimal for coalition $[0,1]$. We conclude that any segmentation $S$ has a segment $(C,p)$ with $p > v_1$.

Take a segmentation $S$ and a segment $(C,p)$ with $p > v_1$. Because $p$ is optimal for $C$, $C$ contains a positive measure of consumers with value $p$. 
Consider a coalition $C'$ that consists of a positive measure of consumers with value $p$ from $C$ and a positive measure of consumers with value $v_1$ (from any segment).  If $f^{C'}(p)$ is small enough relative to $f^{C'}(v_1)$, then price $v_1$ is optimal for $C'$, so $(C',v_1)$ is a segment.  The surplus of value $p$ consumers in $C'$ is $p - v_1>0$, whereas their surplus in $S$ is zero. The surplus of consumers with value $v_1$ in any segment is zero.  Therefore, segment $(C',v_1)$ objects to $S$,  so $S$ is not in the core.  Since $S$ was any segmentation, the core is empty.
\end{proof}

Even though the core may be empty, we saw that stable segmentations always exist. However, it is not immediately obvious that stability is a less demanding notion than the core.  It is in principle possible that for a segmentation $S$ in the core there is another segmentation $S'$ that is not surplus-equivalent to $S$ such that neither segmentation contains a segment that objects to the other segmentation. 
We show that stability in fact generalizes the the core by showing that the two solution concepts coincide when the core is non-empty.

\begin{proposition}\label{prop: core empty then = stable}
If the core is non-empty, then it is equal to the set of all stable segmentations.  
\end{proposition}
\begin{proof}
Recall from \autoref{thm: core iff efficient single product} that if the core is non-empty, then it consists of the trivial segmentation $\{([0,1],v_1)\}$ and all its surplus-equivalent segmentations.  These segmentations are clearly stable because any such segmentation $S$ only contains segments of the form $(C,v_1)$, so in any non-surplus-equivalent segmentation a positive measure of consumers are offered a price higher than $v_1$, and then there is a segment in $S$ that objects to the other segmentation. Any segmentation that is not surplus equivalent to the trivial segmentation is not in the core because $v_1$ is the lowest price that any consumer faces in any segmentation, so no segmentation blocks the trivial segmentation.
\end{proof}

\subsection{Relationship to the stable set}
Our notion of stability is related to the notion of a \emph{stable set} from \cite{MoV53}.  The stable set is defined for any cooperative game; we present its application to our game.  Notice that whereas the stability notion of \cite{MoV53}, stated below, is a property of a \emph{set} of segmentations, our notion of stability is a property of a single segmentation. 

\begin{definition}[Stable set, \citealp{MoV53}]
		A \emph{set} of segmentations $\mathcal{S}$ is a {stable set} if it satisfies the following two properties:		\begin{enumerate}
				\item {Internal Stability:} For any $S \in \mathcal{S}$, no $S' \in \mathcal{S}$ blocks $S$.
				\item {External Stability:} For any $S \notin \mathcal{S}$, some $S'  \in \mathcal{S}$ blocks $S$.
			\end{enumerate}
\end{definition}

If a segmentation $S$ is stable, then the set of all segmentations that are surplus-equivalent to $S$ is a stable set.  This is easy to see: internal stability is trivially satisfied because a segmentation does not block a surplus-equivalent segmentation, and external stability is satisfied by definition of stability.  Because stable segmentations always exist, stable sets exist in our setting. This is noteworthy because stable sets do not exist for some cooperative games. Moreover, even when stable sets exist, they may necessarily contain multiple elements.  In contrast, \autoref{prop: stable sets} shows that any stable set in our setting contains an essentially unique element in the sense that it consists of all segmentations that are surplus-equivalent to some segmentation $S$.

We point out that the set of segmentations that are surplus-equivalent to $S$ may be a stable set even if $S$ is not stable.  This is because stability requires that a \emph{single} segmentation block any other non-surplus-equivalent segmentation; for the set of segmentations that are surplus-equivalent to $S$ to be a stable set, on the other hand, requires that any segmentation that is not surplus-equivalent to $S$ be blocked by \emph{some} segmentation that is surplus-equivalent to $S$. It may be that $S$ does not block $S'$ but a segmentation that is surplus-equivalent to $S$ does.  To see this, suppose that $S$ is canonical but not stable, and consider another segmentation $S'$ that is not blocked by $S$.  Take a segment $(C,p)$ in $S$.  Since $(C,p)$ does not object to $S'$, $C$ may contain some consumers who prefer $S$ to $S'$ and some consumers who prefer $S'$ to $S$.  If coalition $C'\subseteq C$ is such that $(C',p)$ and $(C \backslash C',p)$ are segments and we replace $(C,p)$ with $(C',p)$ and $(C \backslash C',p)$, it could be that $(C',p)$ objects to $S'$, yielding a segmentation that is surplus-equivalent to $S$ and blocks $S'$.  The following example illustrates this.

\begin{example}\label{example 6}
        There are three values, $1,2,3$, with measures $\tfrac{6}{21},\tfrac{4}{21},\tfrac{11}{21}$, respectively, as shown in \autoref{fig:example 6}, and a segmentation $S = \{(C_1,1),(C_2,2)\}$ with $C_1 = [0,\tfrac{6}{21}) \cup [\tfrac{18}{21},1]$  and $C_2 = [\tfrac{6}{21},\tfrac{18}{21})$.
\begin{figure}
    \centering
    		\begin{tikzpicture}[scale=3.5, ultra thick]

	\draw (0,0) node[left]{Consumers} -- (1,0);
	\draw (0,.05) -- (0,-0.05) node[below]{$0$};
	\draw (1,.05) -- (1,-0.05) node[below]{$1$};
	\draw (0,.25) node[left]{Values};
	
	\draw[loosely dotted, thick] (6/21,0) node[below, yshift=-3]{\footnotesize $\frac{6}{21}$} -- (6/21,.25);
	\draw[loosely dotted, thick] (10/21,0) node[below, yshift=-3]{\footnotesize $\frac{10}{21}$} -- (10/21,.25);
	
	\draw (3/21,0.25) node{$1$};
	\draw (8/21,0.25) node{$2$};
	\draw (15.5/21,0.25) node{$3$};
	
	
	\draw (0,-.3) node[left]{$C_1$};
	\draw (0,-.3) -- (6/21,-.3);
	\draw (0,-.25) -- (0,-0.35) node[below]{$0$};
	\draw (6/21-0.01,-.255) arc (40:-40:0.07) node[below]{\footnotesize $\frac{6}{21}$};
	\draw (18/21,-.25) -- (18/21,-.35) node[below]{\footnotesize $\frac{18}{21}$};
	\draw (18/21,-.3) -- (1,-.3);
	\draw (1,-.25) -- (1,-.35) node[below]{$1$};	
	
	\draw (0,-.6) node[left]{$C_2$};
	\draw (6/21,-.6) -- (18/21,-.6);
	\draw (6/21,-.55) -- (6/21,-0.65) node[below]{\footnotesize $\frac{6}{21}$};
	\draw (18/21-0.01,-.555) arc (40:-40:0.07) node[below]{\footnotesize $\frac{18}{21}$};

	\draw (0,-.9) node[left]{$C'_1$};
	\draw (0,-.9) -- (7/21,-.9);
	\draw (0,-.85) -- (0,-0.95) node[below]{$0$};
	\draw (7/21-0.01,-.855) arc (40:-40:0.07) node[below]{\footnotesize $\frac{7}{21}$};
	\draw (18/21,-.85) -- (18/21,-.95) node[below]{\footnotesize $\frac{18}{21}$};
	\draw (18/21,-.9) -- (1,-.9);
	\draw (1,-.85) -- (1,-.95) node[below]{$1$};	
	
	\draw (0,-1.2) node[left]{$C'_2$};
	\draw (7/21,-1.2) -- (18/21,-1.2);
	\draw (7/21,-1.15) -- (7/21,-1.25) node[below]{\footnotesize $\frac{7}{21}$};
	\draw (18/21-0.01,-1.155) arc (40:-40:0.07) node[below]{\footnotesize $\frac{18}{21}$};
	
		\draw (0,-1.5) node[left]{$C''_1$};
	\draw (0,-1.5) -- (6/21,-1.5);
	\draw (0,-1.45) -- (0,-1.55) node[below]{$0$};
	\draw (6/21-0.01,-1.455) arc (40:-40:0.07) node[below]{\footnotesize $\frac{6}{21}$};
	\draw (18/21,-1.45) -- (18/21,-1.55) node[below]{\footnotesize $\frac{18}{21}$};
	\draw (18/21,-1.5) -- (1,-1.5);
	\draw (1,-1.45) -- (1,-1.55) node[below]{$1$};	
	
	\draw (0,-1.8) node[left]{$C''_2$};
	\draw (6/21,-1.8) -- (7/21,-1.8);
	\draw (6/21,-1.75) -- (6/21,-1.85) node[below,xshift=-3.5]{\footnotesize $\frac{6}{21}$};
	\draw (7/21-0.01,-1.755) arc (40:-40:0.07) node[below,xshift=3.5,yshift=-.5]{\footnotesize $\frac{7}{21}$};
	
    \draw (16/21,-1.8) -- (18/21,-1.8);
	\draw (16/21,-1.75) -- (16/21,-1.85) node[below,xshift=-3]{\footnotesize $\frac{16}{21}$};
	\draw (18/21-0.01,-1.755) arc (40:-40:0.07) node[below,xshift=3,yshift=-.5]{\footnotesize $\frac{18}{21}$};
	
	\draw (0,-2.1) node[left]{$C'''_2$};
	\draw (7/21,-2.1) -- (16/21,-2.1);
	\draw (7/21,-2.05) -- (7/21,-2.15) node[below]{\footnotesize $\frac{7}{21}$};
	\draw (16/21-0.01,-2.055) arc (40:-40:0.07) node[below,yshift=-.5]{\footnotesize $\frac{16}{21}$};	

		\end{tikzpicture}
		\caption{\autoref{example 6}}
		\label{fig:example 6}
\end{figure}
\end{example}
Segmentation $S$ is not stable because it is not saturated.  This is because we can add some consumers with value $2$ from $C_2$ to $C_1$ without increasing the price $p=1$ in the first segment.  It is also easy to see directly that $S$ is not stable. For example, segmentation $S' = \{(C'_1,1),(C'_2,3)\}$ with coalitions $C'_1 = [0,\tfrac{7}{21}) \cup [\tfrac{18}{21},1]$  and $C'_2 = [\tfrac{7}{21},\tfrac{18}{21})$ shown in \autoref{fig:example 6} is not blocked by $S$.  Segment $(C_1,1)$ in $S$ does not object to $S'$ because all consumers in $C_1$ are indifferent between the two segmentations.  Segment $(C_2,2)$ does not object to $S'$ because the value 2 consumers who join the first segment in $S'$ strictly prefer $S'$ to $S$. However, segmentation $S'$ is blocked by segmentation $S'' = \{(C''_1,1),(C''_2,2),(C'''_2,2)\}$ with coalitions $C''_1 = [0,\tfrac{6}{21}) \cup [\tfrac{18}{21},1]$, $C''_2 = [\tfrac{6}{21},\tfrac{7}{21}) \cup [\tfrac{16}{21},\tfrac{18}{21})$, and $C'''_2 = [\tfrac{7}{21},\tfrac{16}{21})$, which is surplus-equivalent to $S$.  In particular, segment $(C'''_2,2)$ objects to $S'$ because the consumers in $C'''$ face price $2$ in $S''$ and price $3$ in $S'$.  

The proposition below characterizes stable sets and shows that in this example the set of all segmentations that are surplus-equivalent to $S$ is in fact a stable set.  To state the proposition, we first define two weak notions of objection and blocking.

\begin{definition}[Weak Objections]
A segment $({C},p)$ weakly objects to a segmentation $S$ if $CS(c,p) > CS(c,S)$ for a positive measure of consumers $c$ in $C$ and $CS(c,p) \geq CS(c,S)$ for a positive measure of consumers $c$ in $C$ whose value $v$ is an optimal price for $C$.
\end{definition}

Any objection is also a weak objection.  To see this, observe that both objections and weak objections require that some consumers strictly prefer the segment to the segmentation.  But objections also require that all consumers in the segment weakly prefer the segment.  Weak objections do not require this for consumers whose value is not an optimal price for the segment.  And for values that are optimal prices for the segment, only some consumers with such values are required to prefer the segment.  We now define the corresponding notion of weak blocking.

\begin{definition}[Weak Blocking]
A segmentation $S$ weakly blocks a segmentation $S'$ if there exists a segment $(C,p)$ in $S$ that weakly objects to $S'$.
\end{definition}

Segmentation $S$ in \autoref{example 6} weakly blocks (but does not block) segmentation $S'$ because segment $(C_2,2)$ weakly objects to $S'$: consumers with value 3 in $C_2$ strictly prefer the segment to $S'$, consumers with value 2 in $C_2\cap C'_2$ weakly prefer the segment to $S'$, and 2 is an optimal price for $C_2$.\footnote{$(C_2,2)$ does not object to $S'$ because consumers in $C_2\backslash C'_2$ face a price of 2 in $(C_2,2)$ and a price of 1 in $S'$.}

The following proposition, whose proof is in \autoref{app: stable sets}, characterizes the stable sets.

\begin{proposition}\label{prop: stable sets}
A set of segmentations $\mathcal{S}$ is a stable set if and only if it comprises all the segmentations that are surplus-equivalent to some segmentation $S$ that weakly blocks any segmentation $S'$ that is not surplus-equivalent to $S$.
\end{proposition}

The canonical segmentation $S$ in \autoref{example 6} weakly blocks any non-surplus-equivalent segmentation, so the set of segmentations that are surplus-equivalent to $S$ is a stable set.  To see that $S$ weakly blocks any non-surplus-equivalent segmentation, consider some segmentation $S'$ that is not weakly blocked by $S$.  Suppose without loss of generality that $S'$ is canonical, so $S' = \{(C'_1,1),(C'_2,2),(C'_3,3)\}$.  Because all consumers with value 1 have zero surplus, if some consumers with value 3 in $C_1$ strictly preferred $S$ to $S'$, then $(C_1,1)$ would weakly object to $S'$.  Thus, the consumers with value 3 in $C_1$ are in $C'_1$.  It is impossible for all value $2$ consumers in $C_2$ to be in $C'_1$, because then the revenue from price 2, $2\cdot (\frac{4}{21}+\frac{11}{21})$, would be strictly higher than the revenue from price 1, which is at most $\frac{6}{21} + \frac{4}{21}+\frac{11}{21}$.  Thus, some consumers with value $2$ in $C_2$ weakly prefer $S$ to $S'$.  This implies that the consumers with value $3$ in $C_2$ must be in $C'_1 \cup C'_2$, otherwise some such consumers, those in $C'_3$, would strictly prefer $S$ to $S'$ and then $(C_2,2)$ would weakly object to $S'$.  Therefore $C'_3$ is empty.  Let $\delta_1\geq 0$ be the measure of value 1 consumers in $C'_2$, and let $\delta_2,\delta_3\geq0$ be the measures of value 2 and value 3 consumers from $C_2$ that are in $C'_1$.  For price 1 to be optimal for $C'_1$, the revenue from this price, $\frac{6}{21} - \delta_1 + \delta_2+\frac{3}{21} + \delta_3$, must be no lower than the revenue from price 3, $3 \cdot (\frac{3}{21}+\delta_3)$, which means that $-\delta_1+\delta_2+\delta_3 \geq 3\delta_3$.  Similarly, for price $2$ to be optimal for $C'_2$ we must have $2\cdot(\frac{4}{21}-\delta_2+\frac{8}{21}-\delta_3) \geq 3\cdot(\frac{8}{21}-\delta_3)$, which means that $3\delta_3 \geq 2(\delta_2+\delta_3)$.  These two inequalities hold if and only if $\delta_1 = \delta_2 = \delta_3 = 0$.  We therefore have that $C_1 = C'_1$ and $C_2 = C'_2$, so $S'$ is surplus-equivalent to $S$.

\subsection{Relationship to farsighted stability}\label{sec: farsighted stability}
We now discuss the connection between our notion of stability and two notions motivated by farsighted stability: the Harsanyi stable set and the Ray and Vohra farsighted stable set (henceforth RV stable set).

Both notions define a stable set as one that satisfies internal and external stability, just like the stable set of \cite{MoV53}. But the notion of blocking used to define internal and external stability is ``farsighted.'' A segmentation blocks another segmentation if there is a sequence of segmentations that begins with the segmentation to be blocked and ends with the blocking segmentation such that each intermediate segmentation contains a coalition that prefers the blocking segmentation to the one that preceded the intermediate segmentation. These objecting coalitions allow the blocking segmentation to be ``reached'' starting from the original segmentation. The two notions differ in what is assumed about the segments along the sequence other than the objecting segments, with the RV stable set assuming a kind of ``coalitional autonomy'' similar to the ``coalitional IR'' that motivates our definition of stability.

Our notion of stability satisfies these two notions which, although differing in general, coincide in our setting. Moreover, although these are set notions, in our setting they are satisfied only by singleton sets. Importantly, however, these notions are not particularly useful in our setting because they are too permissive. More precisely, for each notion we have a weak and a strong version; any segmentation that does not eliminate all consumer surplus satisfies the weak versions, and any Pareto undominated segmentation satisfied the strong versions.  We provide the details in \autoref{app: farsighted stability}.


\section{Conclusions}\label{sec: conclusion}
We study market segmentation of a monopolistic market when consumers know their value for the product prior to the segmentation of the market. Because different consumers rank the possible segmentations differently, it is not clear which market segmentation would arise. Instead of formulating a specific game to capture the interaction among consumers and the seller that determines the segmentation, we develop a notion of stability that captures a segmentation being immune to deviations to other segmentations. A stable segmentation is one that, for each segmentation considered as a possible deviation, contains a coalition of consumers that object to the deviation. This captures a kind of ``coalitional individuals rationality (IR).''

Our main result characterizes stable segmentations as those that are efficient and saturated, in that enlarging any segment by adding consumers who face higher prices necessarily increases the profit-maximizing price for the segment. We use this characterization to show that stable segmentations always exist by showing that a particular segmentation (the MER) that maximizes average consumer surplus, identified by \cite{BBM15}, is stable. We also show that efficiency and maximizing consumer surplus is neither necessary nor sufficient for a segmentation to be stable. We highlight the separate roles that efficiency and saturation play by showing that efficient segmentation are those that are fragmentation proof, in that they are immune to objections by sets of consumers that are subsets of existing coalitions. The relationship between the various notions is illustrated in \autoref{fig:results summary}. Finally, we show that our framework can be formulated as a cooperative NTU game and that our notion of stability satisfies many existing solution concepts. Applied to our framework, these solution concepts are not particularly useful because they are too permissive, that is, a large set of segmentations satisfy them.

Our results indicate that a monopolist's use of consumer data to segment the market could be considered as a policy tool to overcome the loss of efficiency associated with monopoly pricing. While efficiency is also achieved with first-degree price discrimination, our results show that as long as ``coalitional IR'' is maintained, the resulting efficient segmentation is Pareto undominated and may increase consumer surplus up to the highest amount possible in the ``surplus triangle'' of 
 \citet{BBM15}. Thus, monopolistic price discrimination subject to ``coalitional IR'' can be viewed as a possible alternative or addition to standard anti-trust regulation.

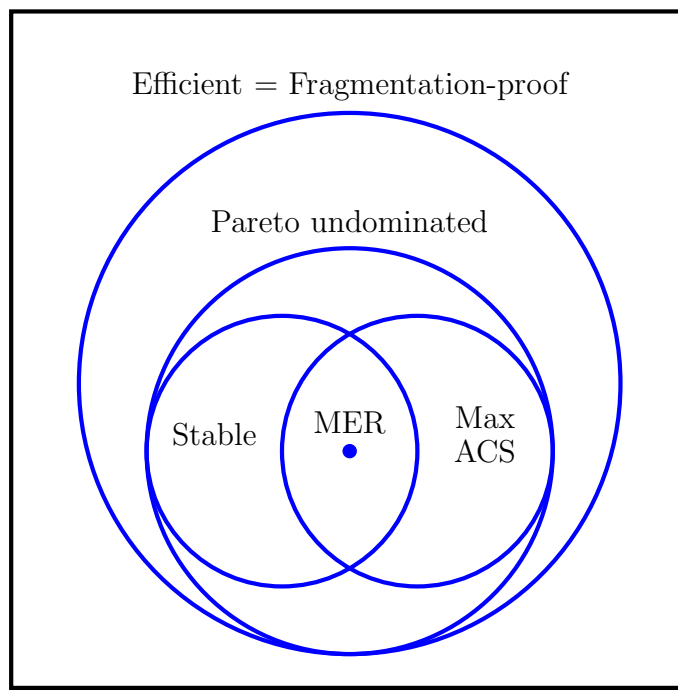
\begin{figure}
    \centering
    	\begin{center}
	\begin{tikzpicture}[scale=4.5, ultra thick]
		\draw (0,0) rectangle (2,2) node[above,xshift=-40]{Segmentations};
		
		\draw[blue] (1,.9) ellipse (.8cm and .8cm);
		\draw (1,1.78) node{Efficient $=$ Fragmentation-proof};
		
		\draw[blue] (1,.7) ellipse (.6cm and .6cm);
		\draw (1,1.38) node{Pareto undominated};
		
		\draw[blue] (.8,.7) ellipse (.4cm and .4cm);
		\draw (.6,.75) node{Stable};

		\draw[blue] (1.2,.7) ellipse (.4cm and .4cm);
		\draw (1.4,.8) node{Max};			
		\draw (1.4,.7) node{ACS};			
		
		\draw[blue, fill = blue] (1,.7) ellipse (.015cm and .015cm)  node[above,black,yshift=2]{MER};
	\end{tikzpicture}
\end{center}
    \caption{Summary of the relationship between different notions.}
    \label{fig:results summary}
\end{figure}


\appendix
\section{Appendix}
\subsection{Proof of \autoref{prop: stable sets}}\label{app: stable sets}
\begin{proof}
To see the necessity of these conditions, consider a stable set $\mathcal{S}$ of segmentations.  We first show that any segmentation in $\mathcal{S}$ is Pareto undominated.  Suppose for contradiction that a segmentation $S$ in $\mathcal{S}$ is Pareto dominated by another segmentation $S'$.  If $S'$ is in $\mathcal{S}$, then internal stability is violated because $S'$ blocks $S$.  If $S'$ is not in $\mathcal{S}$, then, by external stability, there is a segmentation $S''$ in $\mathcal{S}$ that blocks $S'$.  But then $S''$ also blocks $S$, which violates internal stability.  Pareto undominance implies that any segmentation in $\mathcal{S}$ is efficient.

We now show that any two segmentations in $\mathcal{S}$ are surplus-equivalent.  Suppose for contradiction that segmentations $S_1$ and $S_2$ in $\mathcal{S}$ are not surplus-equivalent.  Their induced canonical segmentations $S'_1$ and $S'_2$ are also not surplus-equivalent, so there is a price $p$ and segments $(C_1,p)$ in $S'_1$ and $(C_2,p)$ in $S'_2$ with $C_1 \neq C_2$ (in the ``almost all" sense), where $C_1$ or $C_2$ may be empty.  Suppose without loss of generality that $p$ is the lowest such price, so any consumer in $C_1$ is either in $C_2$ or in a segment of $S'_2$ with a higher price, and similarly any consumer in $C_2$ is either in $C_1$ or in a segment of $S'_1$ with a higher price (up to a set of consumers of measure $0$).  Because $C_1 \neq C_2$, either $C_1 \backslash C_2 $ or $C_2 \backslash C_1$ has positive measure.  Suppose without loss of generality that $C_1 \backslash C_2$ has positive measure.

First observe that $C_1 \backslash C_2$ cannot contain a positive measure of consumers with value $p$. Indeed, such consumers would be in segments of $S'_2$ with prices strictly higher than $p$, so $S'_2$ would not be efficient, contradicting the efficiency of $S_2$. Therefore, $C_1 \backslash C_2$ contains a positive measure of consumers with values higher than $p$.


Consider any segment $(C'',p)$ in $S_1$ that contains some such consumers, that is, $C''\cap(C_1 \backslash C_2)$ has positive measure. Because $C''\subseteq C_1$, the consumers in $C''$ face prices no lower than $p$ in $S_2$, and the consumers in $C''\cap(C_1 \backslash C_2)$ face prices strictly higher than $p$ in $S_2$. So $S_1$ blocks $S_2$, which contradicts internal stability. 

We have established that $\mathcal{S}$ may only contain segmentations that are surplus-equivalent to a Pareto undominated segmentation $S$.  If some $S'$ that is surplus-equivalent to $S$ is not in $\mathcal{S}$, then no segmentation in $\mathcal{S}$ blocks $S'$ so external stability is violated.  So $\mathcal{S}$ must contain \emph{all} segmentations that are surplus-equivalent to a Pareto undominated segmentation $S$, which we can assume to be canonical without loss of generality. To complete the necessity direction, it remains to show that the canonical segmentation $S$ weakly blocks any non-surplus-equivalent segmentation.

Suppose for contradiction that $S$ does not weakly block some non-surplus-equivalent segmentation $S'$.  Because $\mathcal{S}$ is a stable set and contains all segmentations that are surplus-equivalent to $S$, there is a segmentation $S''$ that blocks $S'$ and is surplus-equivalent to $S$.  Consider a segment $(C'',p)$ in $S''$ that objects to $S'$, and the unique segment $(C,p)$ in $S$ in which the price is $p$.   Because $(C'',p)$ objects to $S'$ and $C'' \subseteq C$, there is a positive measure of consumers in $(C,p)$ that strictly prefer $S$ to $S'$.  Because $S$ does not weakly block $S'$, there exists some optimal price $v$ for $C$ such that all consumers with value $v$ in $C$ strictly prefer $S'$ to $(C,p)$.  We claim that $v$ is also optimal for any segment in $S''$ with price $p$, and therefore for $C''$.  To see this, consider all the segments $(C''_1,p),\ldots,(C''_k,p)$ in $S''$ with price is $p$, so $C$ is the union of all these coalitions, one of which is $C''$.  Because $p$ is optimal for $C''_j$, $j=1,\ldots,k$,  we have $vF^{C''_j}(v) \leq pF^{C''_j}(p)$.  If price $v$ is not optimal for some $C''_j$, then $vF^{C''_j}(v) < pF^{C''_j}(p)$. In this case, summing up over all $j$, we have $vF^{C}(v) < pF^{C}(p)$, which contradicts the optimality of price $p$ for coalition $C$. So $v$ must be optimal for $C''$.  Therefore, $C''$ contains a positive measure of consumers with value $v$.  And because all consumers with value $v$ in $C$ strictly prefer $S'$ to $S$, and $S''$ is surplus-equivalent to $S'$, all these consumers strictly prefer $S'$ to $(C'',p)$ so $(C'',p)$ cannot object to $S'$, a contradiction.


To establish sufficiency, consider any canonical segmentation $S = \{(C_1,v_1), \ldots, (C_n,v_n)\}$ that weakly blocks any non-surplus-equivalent segmentation.  The set of segmentations that are surplus-equivalent to $S$ satisfies internal stability because no segmentation blocks a surplus-equivalent segmentation.  For external stability, we show that for any segmentation $S'$ that is not surplus-equivalent to $S$, there is a segmentation $S''$ that is surplus-equivalent to $S$ and blocks $S'$.

Consider the segment $(C,p)$ in $S$ that weakly objects to $S'$.  We will construct a coalition $C'' \subseteq C$ and show that the segmentation $S''$ that is the same as $S$ except that $(C,p)$ is replaced with $(C \backslash C'', p)$ and $(C'',p)$, and is therefore surplus-equivalent to $S$, objects to $S'$.  The construction of $C''$ has two steps.  First, let $C''_1$ be a small coalition that comprises consumers with all the values that are optimal prices for $C$ in proportions that make these values optimal prices for $C''_1$.  That is, $\epsilon = vF^{C''_1}(v)$ for some small $\epsilon$ and all $v$ that are optimal prices for $C$.  For the second step, let $v'$ be such that a positive measure of consumers in $C$ with value $v'$ strictly prefer $(C,p)$ to $S'$.  For some $\delta>0$, add to $C''_1$ a measure $\delta$ of consumers in $C$ with value $v'$ that strictly prefer $(C,p)$ to $S'$, and remove from $C''_1$ the same measure $\delta$ of consumers with the highest value in $C''_1$ that is at most $v'$ (a positive measure of these consumers exists because some consumers in $C''_1$ have value $p$ and $p < v'$, otherwise consumers with value $v'$ have zero surplus in $S$ so do not strictly prefer $(C,p)$ to $S'$).  The resulting coalition is $C''$, which, if $\delta$ is small relative to $\epsilon$, satisfies that $\epsilon = vF^{C''}(v)$ for all prices $v$ that are optimal for $C$.  So if $\delta$ is small relative to $\epsilon$, then $C''$ has the same set of optimal prices as $C''_1$, and $(C'',p)$ is a segment.  Similarly, if $\epsilon$ and $\delta$ are small enough, then $C \backslash C''$ has the same set of optimal prices as $C$, so $(C \backslash C'', p)$ is a segment.  By construction, $(C'',p)$ objects to $S'$, so $S''$, which is surplus-equivalent to $S$, blocks $S'$.
\end{proof}

\subsection{Appendix for \autoref{sec: farsighted stability}}\label{app: farsighted stability}
To apply the Harsanyi and RV stable sets to our setting, we need to address two technical issues.
First, these notions are defined for a finite number of players.  Second, they involve a definition of objection that requires a strict improvement for all members of the objecting coalition.  In our setting, consumers with the lowest value in a coalition have zero surplus, so these notions become trivial (every segmentation satisfies them) if we require a strict improvement for every consumer.  We define modified versions of these notions below, allowing for a continuum of players and weak improvements.  Because farsighted stability considers sequences of deviations, there are two ways to allow for weak improvements.  We therefore define two versions of each solution concept.

\begin{definition}
   A segmentation $S$ Harsanyi blocks a segmentation $S'$ if there is a sequence $S^0 = S', S^1, \ldots, S^n = S$ of segmentations and a sequence $(C^1,p^1),\ldots,(C^n,p^n)$ of segments such that for $i=1 \ldots n$, $(C^i,p^i) \in S^i$ and $CS(c,S^{i-1}) \leq CS(c,S)$ for all consumers $c \in C^i$, with a strict inequality for a positive measure of consumers $c \in C^i$ for some $i$.  If, in addition, $CS(c,S^{i-1}) < CS(c,S)$ for a positive measure of consumers $c \in C^i$ for \emph{all} $i=1 \ldots n$, we say that $S$ strongly Harsanyi blocks $S'$.
\end{definition}

\begin{definition}
		A \emph{set} of segmentations $\mathcal{S}$ is a (strong) {Harsanyi stable set} if it satisfies the following two properties:		\begin{enumerate}
				\item {Internal Stability:} For all $S \in \mathcal{S}$, there exists no $S' \in \mathcal{S}$ that (strong) Harsanyi blocks $S$.
				\item {External Stability:} For all $S \notin \mathcal{S}$, there exists $S'  \in \mathcal{S}$ that (strong) Harsanyi blocks $S$.
			\end{enumerate}
\end{definition}

\begin{definition}\label{def: RV blocking}
   A segmentation $S$ RV blocks a segmentation $S'$ if there is a sequence $S^0 = S', S^1, \ldots, S^n = S$ of segmentations and a sequence $(C^1,p^1),\ldots,(C^n,p^n)$ of segments such that for $i=1 \ldots n$, $(C^i,p^i) \in S^i$ and $(C,p) \in S^i$ whenever $(C,p) \in S^{i-1}$ and $C \cap C^i = \emptyset$, and $CS(c,S^{i-1}) \leq CS(c,S)$ for all consumers $c \in C^i$, with a strict inequality for a positive measure of consumers $c \in C^i$ for some $i$.  If, in addition, $CS(c,S^{i-1}) < CS(c,S)$ for a positive measure of consumers $c \in C^i$ for \emph{all} $i=1 \ldots n$, we say that $S$ strongly RV blocks $S'$.
\end{definition}

\begin{definition}
		A \emph{set} of segmentations $\mathcal{S}$ is a (strong) {RV stable set} if it satisfies the following two properties:		\begin{enumerate}
				\item {Internal Stability:} For all $S \in \mathcal{S}$, there exists no $S' \in \mathcal{S}$ that (strong) RV blocks $S$.
				\item {External Stability:} For all $S \notin \mathcal{S}$, there exists $S'  \in \mathcal{S}$ that (strong) RV blocks $S$.
			\end{enumerate}
\end{definition}

For the following characterization of Harsanyi and RV stable sets we denote by $ACS(S)$ the average consumer surplus in segmentation $S$.

\begin{proposition}\label{prop: characterizing HSS and RVSS}
    The following are equivalent for any set of segmentations $\mathcal{S}$:
    \begin{itemize}
        \item $\mathcal{S}$ is a Harsanyi stable set
        \item $\mathcal{S}$ is a RV stable set
        \item $\mathcal{S} = \{S\}$ for some $S$ with $ACS(S) > 0$.
    \end{itemize}
\end{proposition}

The proof of \autoref{prop: characterizing HSS and RVSS} uses the following lemma.

\begin{lemma}\label{lem: characterizing H and RV blocking}
For any two segmentations $S$ and $S'$, the following are equivalent:
\begin{itemize}
    \item $S$ Harsanyi blocks $S'$.
    \item $S$ RV blocks $S'$.
    \item $ACS(S)>0$.
\end{itemize}
\end{lemma}
\begin{proof}
If $ACS(S)=0$, then $CS(c,S)=0$ for all consumers.  Therefore, $S$ cannot Harsanyi block or RV block any segmentation.

Suppose that $ACS(S)>0$.  We show that $S$ RV blocks any segmentation $S'$, which also implies that $S$ Harsanyi blocks $S'$.  We do so by constructing a sequence of segmentations in several steps that gradually transform $S'$ to a segmentation in which each segment includes consumers with a single value. We then proceed from the elementary segmentation to $S$. 

In step $0$ we set $S^0=S'$. In each following step $i>0$, we take the segmentation $S^{i-1}$ and a segment $(C,p)$ in $S^{i-1}$ that contains consumers of at least two types. For each value $v_j$, we let $C^i_{v_j}$ be the set of all consumers with value $v_j$ in $C$.  $S^i$ is constructed from $S^{i-1}$ by replacing $(C,p)$ with the segments $(C^i_{v_j},v_j)$ for all $j$ such that $f^C(v_j)>0$.  Let $C^i = C^i_p$.  The first phase ends with a segmentation in which every segment contains consumers of only a single type, so the surplus of all consumers is zero. The next segmentation in the sequence is $S$, which completes the construction of the sequence.


To see that $S$ RV blocks $S'$, notice that in each step $i=1,\ldots,n$, consumers in $C^i$ have zero surplus in $S^{i-1}$.  Therefore, they weakly prefer $S$ to $S^{i-1}$.
Additionally, because $ACS(S)> 0$, there is a segment $(C,p)$ in $S$ in which a positive measure of consumers obtain positive surplus.  As a result, a positive measure of consumers strictly prefer $S=S^n$ to $S^{n-1}$.
\end{proof}

\begin{proof}[Proof of \autoref{prop: characterizing HSS and RVSS}]
Suppose that $\mathcal{S} = \{S\}$ for some $S$ with $ACS(S) > 0$.  Then, by \autoref{lem: characterizing H and RV blocking}, $S$ RV blocks and Harsanyi blocks any $S' \neq S$, so $\mathcal{S}$ is a RV stable set and a Harsanyi stable set.  

Consider any Harsanyi (respectively RV) stable set $\mathcal{S}$.  The set $\mathcal{S}$ must contain at least one segmentation $S$ with $ACS(S) > 0$, otherwise a segmentation $S' \notin \mathcal{S}$ is not Harsanyi (RV) blocked by any segmentation in $\mathcal{S}$ by \autoref{lem: characterizing H and RV blocking}.  If the set contains more than one segmentation, then, by \autoref{lem: characterizing H and RV blocking}, the segmentation $S$ that satisfies $ACS(S)>0$ Harsanyi (RV) blocks the other segmentations in the set.  Therefore, $\mathcal{S}$ contains a single segmentation $S$, and $ACS(S)>0$.
\end{proof}

\begin{proposition}\label{prop: characterizing strong HSS and RVSS}
    The following are equivalent for any set of segmentations $\mathcal{S}$:
    \begin{itemize}
        \item $\mathcal{S}$ is a strong Harsanyi stable set
        \item $\mathcal{S}$ is a strong RV stable set
        \item $\mathcal{S}$ is the set of all segmentations that are surplus-equivalent to some segmentation $S$ that is Pareto undominated.
    \end{itemize}
\end{proposition}

The proof uses the following lemma.

\begin{lemma}\label{lem: characterizing strong H and RV blocking}
For any two segmentations $S$ and $S'$, the following are equivalent:
\begin{itemize}
    \item Some surplus-equivalent segmentation to $S$ strong Harsanyi blocks $S'$.
    \item Some surplus-equivalent segmentation to  $S$ strong RV blocks $S'$.
    \item There exist a positive measure of consumers $c$ such that $CS(c,S)>CS(c,S')$.
\end{itemize}
\end{lemma}
\begin{proof}
If some segmentation $S''$ that is surplus-equivalent to $S$ strong Harsanyi (RV) blocks $S'$, then, by definition, a positive measure of consumers strictly prefer $S''$, and therefore $S$, to $S'$.

Suppose that a positive measure of consumers strictly prefer $S$ to $S'$.
We show that some segmentation $S''$ that is surplus-equivalent to $S$ strong RV blocks segmentation $S'$, which also implies that $S''$ strong Harsanyi blocks $S'$.

We do so by constructing a sequence of segmentations in two phases. The first phase consists of two steps.  In the first step, consider some segment $(C,p)$ in $S$ that contains a positive measure of consumers that strictly prefer $S$ to $S'$.  Let coalition $C^1$ contain a positive measure of consumers with value $p$ from $C$, a positive measure of (but not all the) consumers from $C$ that strictly prefer $S$ to $S'$, and, for \emph{every} segment in $S'$, a positive measure of consumers with the lowest value in that segment, where the proportions of consumers in $C^1$ are such that $(C^1,p)$ is a segment.  Consider a segmentation $S^1$ that consist of $(C^1,p)$ and, for each consumer value, a segment that contains only the consumers in $[0,1]\backslash C^1$ with that value, so their surplus is zero.  In the second step, replace $(C^1,p)$ with $(C,p)$ and, for each consumer value, put the consumers in $C^1 \backslash C$ with that value in a separate segment (all other segments remain intact). Denote the resulting segmentation by $S^0$.

The second phase consists of (potentially) several steps. In each step $i>0$, take segmentation $S^{i-1}$ and, for some segment $(C',p')$ in $S$ that is not already in $S^{i-1}$ and contains a positive measure of consumers with positive surplus, let $C^i = C'$.  $S^i$ is constructed from $S^{i-1}$ by taking all segments $(C'',p'')$ that contain a positive measure of consumers from $C'$ (so $C''$ contains only consumers with value $p''$) and replacing them with $(C'' \backslash C',p'')$, and finally adding segment $(C',p')$ to $S^i$.  This process ends with a final segmentation $S^n$ that may differ from $S$ but is surplus-equivalent to it because for any segment in $S$ that is not in $S^n$, all consumers in that segment obtain zero surplus in both segmentations.  So, for the remainder of the proof, suppose without loss of generality that $S = S^n$.

To see that $S$ RV blocks $S'$, notice that in the first step of the first phase, coalition $C^1$ contains some consumers that strictly prefer $S$ to $S'$, and all other consumers in $C^1$ weakly prefer $S$ to $S'$ because they ave surplus zero in $S'$.  In the second step of the first phase, by definition, some consumers in $C$ strictly prefer $S$ to $S^1$ and all the consumers in $C \backslash C^1$ because they have zero surplus in $S^1$.  Similarly, in each step $i$ of the second phase, consumers in $C^i$ have surplus zero in $S^{i-1}$, and some consumers in $C^i$ strictly prefer $S$ to $S^{i-1}$ because they have a positive surplus in $S$.
\end{proof}

\begin{proof}[Proof of \autoref{prop: characterizing strong HSS and RVSS}]
Suppose that $\mathcal{S}$ is the set of all segmentations that are surplus-equivalent to some Pareto undominated segmentation $S$.  Any segmentation $S'$ that is not surplus-equivalent to $S$ does not Pareto dominate $S$. By \autoref{lem: characterizing strong H and RV blocking}, $S$ strong RV blocks and strong Harsanyi blocks any such $S'$, and therefore $\mathcal{S}$ is a strong RV stable set and a strong Harsanyi stable set.  

Consider a strong Harsanyi (RV) stable set $\mathcal{S}$.  If the set contains two segmentations $S$ and $S'$ that are not surplus equivalent, then either $S$ is not Pareto dominated by $S'$ or $S'$ is not Pareto dominated by $S$.  Then, by \autoref{lem: characterizing strong H and RV blocking}, one of the two segmentations strong Harsanyi (RV) blocks the other one, violating internal stability.  Therefore, $\mathcal{S}$ contains only surplus-equivalent segmentations. A segmentation $S$ in $\mathcal{S}$ cannot be Pareto dominated by any segmentation $S'$ not in in $\mathcal{S}$ because otherwise, by \autoref{lem: characterizing H and RV blocking}, $S$ would not strong Harsany (RV) block $S'$, violating external stability.\end{proof}

All mentioned notions of stability require a strict improvement in every step of the process, which is different than the two notions we use.  If we require strict improvement in every step, then no segmentation can block any other segmentation because in any segmentation some consumers get zero surplus.  So in that case, the unique RV stable set (and also Harsanyi and also maximal RV stable set) is the set of all segmentations.  RV characterize \emph{singleton} stable sets in their setting.  Our setting is different than theirs because we have a continuum of agents and there are technical assumptions (like comprehensiveness) that our setting does not satisfy.  Nevertheless, our result that the only stable set is the set of all segmentations does not contradict their result that characterizes singleton stable sets.  Our observation implies that there are no such stable sets.  They show that the segmentation in a singleton stable set must be separable.  In our setting, there is no separable segmentation.  For this, say a set of coalitions $C_1,\ldots,C_n$ is a sub-partition if the coalitions are mutually disjoint and their union is a strict subset of the set of all consumers $[0,1]$.  A segmentation $S$ is separable if it is Pareto undominated and for any sub-partition $C_1,\ldots,C_n$ such that for each $i$ there is a segment $(C_i,p_i)$ such that the surplus of all consumers in $C_i$ is the same in $(C_i,p_i)$ and $S$, there is a segment $(C,p)$ such that $C$ is mutually disjoint from $C_1,\ldots,C_n$ and the surplus of all consumers in $C$ is the same in $(C,p)$ and $S$.  To see that there are no separable segmentation in our setting, consider any segmentation $S$.  For each value $v_i$, let $C_i$ be the set of all value $v_i$ consumers that have zero surplus in $S$.  Notice that by definition, the surplus of all consumers in $C_i$ is the same in segment $(C_i,v_i)$ and $S$.  Note that $C_1,\ldots,C_n$ is a sub-partition unless the surplus of all consumers is zero, in which case $S$ is Pareto dominated and is therefore not separable.  Also notice that the surplus of all consumers not in sub-partition $C_1,\ldots,C_n$ is positive.  Because for any segment $(C,p)$ the surplus of some consumers are zero, there is no coalition $(C,p)$ where $C$ is disjoint from the sub-partition $C_1,\ldots,C_n$ such that the surplus of all consumers in $C$ is the same in $(C,p)$ and $S$.  Therefore, $S$ is not separable.

\citet{RaV19} define a notion of \emph{maximality} of a stable set and show that any single-payoff RV stable set is also a maximal RV stable set.  Because both \autoref{prop: characterizing HSS and RVSS} and \autoref{prop: characterizing strong HSS and RVSS} characterize stable sets as ones that contain a single segmentation (or surplus-equivalent ones), those stable sets are also maximal (their setting with a finite number of players and strict improvements is slightly different than ours but the arguments are identical).  Roughly speaking, maximality requires that in a chain of segmentations defined in \autoref{def: RV blocking} that ends in $S$, at each step the move specified by the chain is ``optimal" in the sense that no coalition $C$ has another move that would lead to another segmentation in the stable set that the coalition $C$ prefers to $S$.  If the stable set is a singleton, then all chains necessarily end in the same segmentation, and therefore maximality is trivially satisfied.

{\normalsize  \bibliographystyle{jpe}
	\bibliography{bibs}}

\end{document}